\documentclass[journal]{IEEEtran}

\usepackage[utf8]{inputenc} 
\usepackage[T1]{fontenc}
\usepackage{graphicx}
\usepackage{xcolor}
\usepackage{url}
\usepackage{amssymb}
\usepackage{cite}
\usepackage{bm}
\usepackage{arydshln}
\usepackage[caption=false,font=footnotesize]{subfig}
\usepackage{enumitem}
\usepackage[cmex10]{amsmath} 

\newtheorem{definition}{Definition}
\newtheorem{corollary}{Corollary}
\newtheorem{proposition}{Proposition}
\newtheorem{theorem}{Theorem}
\newtheorem{remark}{Remark}
\newtheorem{lemma}{Lemma}
\newcommand{\FCart}{F_{\rm Cart}}
\begin{document}
\title{Dynamic Coded Distributed Computing with Joint File Placement and Reduce Function Reassignment}




\author{Chenglin Li, Minquan Cheng, and Youlong Wu%
    \thanks{Chenglin Li and Youlong Wu are with the School of Information Science and Technology, ShanghaiTech University, Shanghai 201210, China (e-mail: \{lichl2023,wuyl1\}@shanghaitech.edu.cn).}%
    \thanks{Minquan Cheng is with the Key Laboratory of Education Blockchain and Intelligent Technology, Guangxi Normal University, Guilin 541004, China (e-mail: chengqinshi@hotmail.com).}%
}

\maketitle

\begin{abstract}
Coded distributed computing (CDC) mitigates the Shuffle communication
bottleneck through coded multicast, but existing schemes with predetermined
worker populations do not directly adapt to worker arrivals and disconnections.
To address these dynamics, we propose a dynamic CDC scheme that preserves
surviving workers' caches and Reduce assignments while exploiting the storage
brought by arriving clusters. The model accommodates general MapReduce tasks
under a placement that is partly fixed and partly new.
The scheme uses Cartesian products of cluster-wise placement delivery arrays
(PDAs) for file placement. A delivery-aware greedy algorithm reassigns
abandoned Reduce functions, and a Communication Array allows
workers with no Reduce function to serve as coded transmitters. For arbitrary
feasible disconnections, the scheme's exact communication load is derived.
A converse bound under one-shot communication
shows that, in the absence of disconnections, the scheme lies within
a factor of $2$ of the lower bound under these assumptions when new clusters
arrive, and within a factor of $4$ otherwise.
\end{abstract}
\section{Introduction}
Frameworks such as MapReduce \cite{dean2008mapreduce} and Spark
\cite{zaharia2010spark} process large datasets in parallel by distributing
storage and computation across workers. In MapReduce, workers map locally
stored files into intermediate values (IVs), exchange the required IVs during
Shuffle, which can become a communication bottleneck, and apply Reduce functions to obtain the output results To reduce this overhead, coded distributed computing
(CDC) \cite{li2017fundamental}, inspired by coded caching \cite{6763007},
uses redundant file placement to create coded multicast
opportunities and reduce the communication load. Subsequent work has extended CDC to heterogeneous systems. Woolsey
et al.~\cite{Woolsey'21} developed combinatorial file and function assignments
for workers with different storage and computing capabilities. Song
et al.~\cite{Song'22} jointly designed Reduce function assignment and coded
shuffling for a given file placement, while Wang and Wu~\cite{wang2022coded}
designed coded transmissions with both file placement and function assignment
fixed.

These schemes assume a fixed worker population and do not directly account for changes in worker
availability during a task. Workers may disconnect because of failures, while
new workers may become available after completing other tasks. This paper
considers such changes before the Shuffle phase, with arriving workers grouped
into clusters whose sizes and cache capacities may differ. Workers present at
the start of the task are each assigned a Reduce function, whereas arriving
workers initially have none. A disconnected worker's cached files and IVs
become unavailable, and any Reduce function it leaves behind must be
reassigned. Arriving workers bring storage and communication resources but may
also disconnect. To retain existing data and tasks, we preserve the caches and
Reduce assignments of surviving workers. The design must therefore make use
of arriving resources and reassign abandoned functions without reorganizing
the surviving caches or tasks. The worker set is fixed during Shuffle.

Preserving existing caches under worker changes has precedents in coded
elastic computing \cite{CEC} and dynamic coded caching \cite{Zhang'20, Wu'24}.
Elastic computing adapts computation allocation when workers join or leave,
whereas dynamic coded caching serves file demands from a changing worker set.
In our setting, workers execute general MapReduce functions on input files,
and a disconnected worker can leave a Reduce function that another worker must complete.
Reassignment changes the IV demands, while the Shuffle messages must be formed
from IVs cached at connected workers. We jointly design arriving file placement,
function reassignment, and coded Shuffle delivery to exploit the available
storage and communication resources while preserving surviving caches and
Reduce assignments.
Section~\ref{sec:related} details the distinctions from related schemes.

Three constraints make this joint design difficult. First, coding opportunities
must connect the fixed caches of surviving workers with the newly placed files
at heterogeneous arriving clusters. Second, reassigning a Reduce function
changes a worker's missing IVs and thus the coded messages needed for delivery.
Third, in device-to-device (D2D) Shuffle without a server holding all files,
every coded message must be computable at its sender. A worker with no Reduce
function has no IV demand but can still help transmit using its cached IVs.
Consequently, a delivery construction must account for both the reassigned
demands and the available senders.

To address these challenges, we propose a dynamic CDC scheme that jointly
designs file placement, Reduce function reassignment, and coded Shuffle
transmissions. We use Cartesian products~\cite{wang2023placement} of cluster-wise
placement delivery arrays (PDAs)~\cite{7973185} to assign files to arriving workers
while preserving surviving caches. A delivery-aware greedy algorithm
reassigns abandoned Reduce functions using simulated communication load,
and Communication Arrays derived from the Cartesian products determine the coded transmissions.
The main contributions are as follows:

\begin{itemize}
    \item We propose a dynamic CDC scheme that exploits the storage and computation
    resources of surviving and arriving workers while preserving surviving
    caches and Reduce assignments. Cartesian product placement creates
    coded-multicasting opportunities across clusters with different sizes and
    cache capacities. A Communication Array exploits workers without
    Reduce functions as coded transmitters, and delivery-aware reassignment
    aligns abandoned functions with the available coding opportunities.

    \item Because the problem of dynamic CDC with file allocation for newly joined workers and Reduce function reassignment remains unexplored, we extend an existing scheme as a baseline and compare it with our proposed scheme; the latter significantly outperforms the baseline. Furthermore, our scheme supports heterogeneous settings, such as clusters with different numbers of workers and workers with different storage capacities.

    \item We derive the exact communication load for feasible disconnection sets
    and establish a tight worst-case tolerance of $\sum_{c\in[C]}t_c-1$
    disconnected workers, where $t_c$ is the number of copies of each file in cluster $c$.
    Without disconnections, we establish approximation factors of 2 with arrivals and 4
    without, relative to the best load over all feasible placements
    in the one-shot communication class of Section~\ref{sec:model}.
\end{itemize}

\subsection{Related Work}
\label{sec:related}

\textbf{CDC with a fixed worker set.}
CDC research covers wireless Shuffle \cite{7935426}, master-aided computation
\cite{chen2021coded, chen2023optimality}, and heterogeneous systems
\cite{Woolsey'21, 9448271, xu2021new, reisizadeh2019coded, wang2022coded, Song'22, wu2024coded}.
Closely related works optimize file allocation under heterogeneous resources
\cite{Woolsey'21, 9448271, xu2021new}, jointly design function assignment and
coded shuffling \cite{Song'22}, or handle predetermined file and output-function
assignments \cite{wang2022coded}. These formulations take the worker set as
fixed. Here, worker arrivals and disconnections change both the available cached IVs and
the allocation of Reduce tasks.

\textbf{Elastic computing.}
Elastic computing adapts computation to changing worker availability. The
original coded elastic computing (CEC) framework \cite{CEC} adjusts computation
subsets over MDS-coded data while preserving existing storage. Extensions
consider heterogeneous resources
\cite{Woolsey'CEC, Woolsey'GLOBECOM}, straggler tolerance via Lagrange codes
\cite{Zhong'23}, coded transmission \cite{Zhong'24},
decentralized operation \cite{Huang'24}, and storage and download resilience
\cite{11195520}. Storage preservation and distributed operation are thus shared
features. Our model requires workers to execute general MapReduce functions
on input files, retain surviving Reduce assignments, and reassign abandoned
functions while performing cache-based D2D Shuffle. The cited elastic-computing
works focus on computation allocation and result recovery; the present problem
couples function reassignment with IV delivery among workers.

\textbf{Coded caching and placement delivery arrays.}
Our placement-delivery principle builds on coded caching
\cite{6763007, 6807823} and its PDA representation \cite{7973185}, with subsequent
constructions in \cite{10086693,8437603,8613522}. We use the basic-PDA and
Cartesian product approach of \cite{wang2023placement} to combine clusters
with different cache ratios. Dynamic coded caching also addresses changing
user populations \cite{Zhang'20, Wu'24, 9815628}. In particular,
\cite{Zhang'20, Wu'24} place content at arriving users while retaining existing
caches, so cache preservation already has precedents in this setting.
Coded caching serves file demands; a disconnected user leaves no computation
function to be reassigned. Here, reassigning an abandoned Reduce function
changes a worker's IV demands and hence the delivery design. Moreover, in
centralized coded caching a server holding all files forms the messages,
whereas our Shuffle must use IVs locally available at connected workers.

\subsection{Notation}
\label{sec:notation}
Let $\mathbb N^+$ denote the positive integers. For $n\in\mathbb N^+$, write
$[n]=\{1,\ldots,n\}$, with $[0]=\emptyset$. For integers $a,b$, let
$[a:b]=\{m\in\mathbb Z:a\leq m\leq b\}$; this set is empty if $a>b$.
We use $\mathbf1\{\cdot\}$ for the indicator of a condition. Empty sums
and products equal $0$ and $1$, respectively.
For sets $A_1,\ldots,A_m$, their Cartesian product is
$A_1\times\cdots\times A_m=\{(a_1,\ldots,a_m):a_i\in A_i,\ i\in[m]\}$.

\subsection{Paper Organization}
\label{sec:organization}
Section~\ref{sec:model} presents the system model and preliminaries.
Section~\ref{sec:results} states the main results, followed by numerical
comparisons in Section~\ref{sec:numerical}.
Section~\ref{sec:illustrative} gives an example, and Section~\ref{sec:scheme}
presents the construction and its computational complexity.
Section~\ref{sec:conclusion} concludes the paper.
The remaining proofs are given in the appendices.

\section{System Model}
\label{sec:model}
\subsection{Problem Formulation}
\label{sec:formulation}
\begin{figure}[t]
	\centering
	\includegraphics[width=0.8\columnwidth]{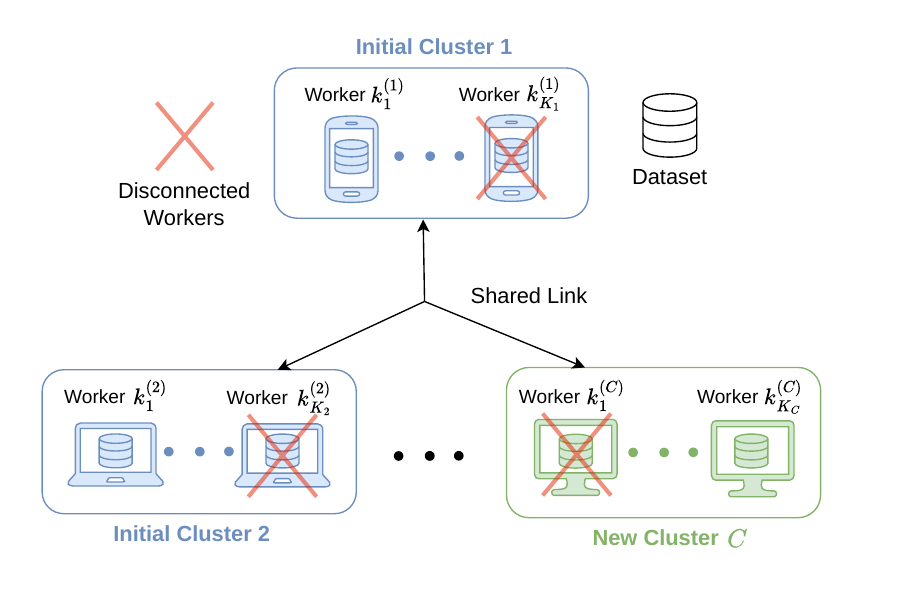}
	\caption{System model with initial and newly arriving clusters.}
	\label{fig:system_model}
\end{figure}
Consider the dynamic MapReduce framework depicted in Fig.~\ref{fig:system_model}.
The system holds $N\in\mathbb N^+$ input files, denoted by $w_1,\ldots,w_{N} \in \mathbb{F}_{2^{F}}$, for some $F\in \mathbb{N}^+$.
The computation specifies $Q\in\mathbb N^+$ output functions, denoted by $\phi_1,\ldots,\phi_{Q}$, where each $\phi_q:(\mathbb{F}_{2^{F}})^{N}\rightarrow \mathbb{F}_{2^{B}}$, $q\in[Q]$, maps the input files $\{w_1,\ldots,w_{N}\}$ to an output result of length $B\in \mathbb{N}^+$.
The system starts with $C_{\textnormal{ini}}\in\mathbb N^+$ initial clusters, indexed by
$\mathcal C_{\rm ini}=[C_{\textnormal{ini}}]$. Before the Shuffle phase,
$C_{\textnormal{new}}\in\mathbb N$ new clusters join and are indexed in
arrival order by
$\mathcal C_{\rm new}=[C_{\textnormal{ini}}+1:C]$, where
$C=C_{\textnormal{ini}}+C_{\textnormal{new}}$ denotes the total number of clusters, and the set $\mathcal C_{\rm new}$ is empty when no cluster joins. The participating cluster set is therefore
$\mathcal C=\mathcal C_{\rm ini}\cup\mathcal C_{\rm new}=[C]$.
Cluster $c\in\mathcal C$ has $K_c\in\mathbb N^+$ workers,
each storing $M_c\in[N]$ files, or $M_cF$ bits. Cluster sizes and cache capacities
may differ. The per-cluster replication level is
$t_c=K_cM_c/N\in\{1,\ldots,K_c-1\}$.
All connected workers communicate over a lossless shared link.
The total number of workers is $K=\sum_{c=1}^{C}K_c$, and they are globally
indexed by $[K]$. Specifically, cluster $c\in[C]$ comprises the workers indexed
by
\begin{IEEEeqnarray}{rCl}
\mathcal K_c
&=&\left[\sum_{n=1}^{c-1}K_n+1:\sum_{n=1}^{c-1}K_n+K_c\right].
\end{IEEEeqnarray}
For $c\in[C]$ and $i\in[K_c]$, we denote the $i$-th worker in cluster
$c$ and its global index by $k_i^{(c)}=\sum_{n=1}^{c-1}K_n+i$.
Thus, the workers of the initial clusters are indexed by $[K_0]$, where
$K_0=\sum_{c\in\mathcal C_{\rm ini}}K_c$. For $c\in[C]$ and
$i\in[K_c]$, let $\mathcal M_i^{(c)}\subseteq[N]$ denote the file
indices cached by worker $k_i^{(c)}$. Each worker in cluster $c$ stores
exactly $M_c$ files, i.e., $|\mathcal M_i^{(c)}|=M_c$ for all $c\in[C]$ and
$i\in[K_c]$.
For each worker $k\in[K]$, let $\mathcal Q_k\subseteq[Q]$ denote the indices
of its assigned Reduce functions. Initially, every initial worker is assigned a distinct Reduce function. Without loss of generality, we assume $\mathcal Q_k=\{k\}$ for
$k\in[K_0]$ and $\mathcal Q_k=\emptyset$ for $k\in[K]\setminus[K_0]$; hence,
$Q=K_0$. A worker with no Reduce function can still transmit coded messages.

Each output function admits the MapReduce decomposition
\begin{IEEEeqnarray}{rCl}
    \phi_q(w_1,\ldots,w_N)=h_q(v_{q,1},\ldots,v_{q,N}),
\end{IEEEeqnarray}
where $v_{q,n}=g_{q,n}(w_n)$ is an intermediate value (IV),
$g_{q,n}:\mathbb F_{2^F}\to\mathbb F_{2^T}$, and
$h_q:(\mathbb F_{2^T})^N\to\mathbb F_{2^B}$, with
$q\in[Q]$, $n\in[N]$, and $T\in\mathbb N^+$.
For communication load analysis, IVs are modeled as independent,
uniformly distributed $T$ bit symbols. Symbol extension is allowed
to ensure integer subpacket lengths, with loads normalized by the
resulting IV length.

\textbf{Map Phase:} Under the placement rule of Section~\ref{sec:placement},
worker $k_i^{(c)}$ maps its cached files to
$\{v_{q,n}:q\in[Q],\ n\in\mathcal M_i^{(c)}\}$ by the map function $g_{q,n}$.

\textbf{Dynamic Cluster Update:} While the workers perform
their Map operations, some workers may disconnect and new clusters may join.
This stage overlaps the Map phase and ends before the Shuffle phase, so the
participating workers are fixed before Shuffle begins. Let
$\mathcal D\subseteq[K]$ be the disconnected workers, partitioned into
disconnected initial and newly joined workers as
\begin{IEEEeqnarray}{rCl}
\mathcal D_{\rm ini}&=&\mathcal D\cap[K_0],\\
\mathcal D_{\rm new}&=&\mathcal D\cap([K]\setminus[K_0]).
\end{IEEEeqnarray}
These sets are disjoint and their union is $\mathcal D$.
Define the worker connection vector $\boldsymbol E=(E_1,\ldots,E_K)\in\{0,1\}^K$ by
\begin{IEEEeqnarray}{rCl}
E_k&=&\begin{cases}1,&k\notin\mathcal{D},\\0,&k\in\mathcal{D},\end{cases}\qquad k\in[K].
\end{IEEEeqnarray}
The workers participating in the Shuffle phase are those in
$\bar{\mathcal D}=\{k\in[K]:E_k=1\}=[K]\setminus\mathcal{D}$.
A disconnected worker no longer contributes its cached IVs or acts as a
transmitter. A worker in $\mathcal D_{\rm ini}$ also abandons its Reduce
function, whereas one in $\mathcal D_{\rm new}$ held none initially. For
$c\in\mathcal C_{\rm ini}$ and $i\in[K_c]$ with
$k_i^{(c)}\in\bar{\mathcal D}$, the cache $\mathcal M_i^{(c)}$ remains
fixed. For $c\in\mathcal C_{\rm new}$ and $i\in[K_c]$ with
$k_i^{(c)}\in\bar{\mathcal D}$, the cache is determined by
Section~\ref{sec:placement}, subject to $|\mathcal M_i^{(c)}|=M_c$.

\textbf{Reduce function Reassignment:} When an initial worker disconnects,
its Reduce function is reassigned to a connected worker using the rule in
Section~\ref{sec:reassignment}. Let
$\gamma_k=|\mathcal Q_k|\in[0:Q]$ for $k\in[K]$, and let
$\boldsymbol\gamma=(\gamma_1,\ldots,\gamma_K)$ denote the Reduce function assignment vector.

\textbf{Shuffle Phase:} Worker $k_i^{(c)}$ needs the IVs
$\{v_{q,n}:q\in\mathcal Q_{k_i^{(c)}},\ n\in[N]\}$ for its assigned functions.
Every connected worker encodes its local IVs using $\psi_i^{(c)}$ into a message
$\mathcal X_i^{(c)}\in\{0,1\}^{\ell_i^{(c)}}$ of length
$\ell_i^{(c)}\in\mathbb{N}$ bits. A disconnected worker sends no message,
so $\ell_i^{(c)}=0$.

\textbf{Reduce Phase:} Each connected worker $k_i^{(c)}$ recovers the IVs
for every $q\in\mathcal Q_{k_i^{(c)}}$ from the Shuffle messages and its local IVs
using a decoding function $\eta^{(c)}_{i,q}$:
\begin{IEEEeqnarray}{rCl}
(v_{q,1},\ldots,v_{q,N})&=&\eta^{(c)}_{i,q}\Bigl(\bigl\{\mathcal{X}^{(c')}_{i'}:c'\in[C],i'\in[K_{c'}]\bigr\},\nonumber\\
&&\bigl\{v_{q',n}:q'\in[Q],n\in\mathcal M_i^{(c)}\bigr\}\Bigr).
\end{IEEEeqnarray}
It then computes $h_q(v_{q,1},\ldots,v_{q,N})$.

\begin{definition}[Communication Load]
\label{def:load}
For a fixed system configuration and placement
$\mathcal M=\{\mathcal M_i^{(c)}:c\in[C],\ i\in[K_c]\}$,
the communication load of a scheme is its total number of bits transmitted
during Shuffle, normalized by $NQT$:
    \begin{IEEEeqnarray}{rCl}
L\bigl(\mathcal C_{\rm ini},\mathcal C_{\rm new},\mathcal{D},\mathcal{M}\bigr)&\triangleq&\frac{\sum_{c\in\mathcal C_{\rm ini}\cup\mathcal C_{\rm new}}\sum_{i\in[K_c]}\ell_i^{(c)}}{NQT}.
    \end{IEEEeqnarray}

\end{definition}

Let $\Omega$ denote the class of one-shot communication schemes in which every message is an XOR
of IV pieces and each receiver decodes one missing piece from that message
and its cached IVs alone; let $\Psi$ denote all possible communication schemes where dependency between messages is allowed, $\Omega \subseteq \Psi$. Let $\mathcal P$ denote all feasible file placements; Let $\mathcal{Q}$ denote all feasible reduce function assignments.
\begin{definition}[Optimal general communication load]
 For given placement $\mathcal M\in \mathcal{P}$ and Reduce function assginment $\boldsymbol\gamma\in \mathcal{Q}$, define the optimal communication load as 
 \begin{IEEEeqnarray}{rCl}
L^*
&\triangleq&\inf\{L:
(\mathsf \psi,L)\textnormal{ is feasible}\},
\end{IEEEeqnarray}
where $\psi\in \Psi$ denotes a general communication scheme.
\end{definition}

\begin{definition}[Optimal one-shot Communication Load]
For given placement $\mathcal M\in \mathcal{P}$ and Reduce function assginment $\boldsymbol\gamma\in \mathcal{Q}$,
define the optimal one-shot communication load as
\begin{IEEEeqnarray}{rCl}
L_{\textnormal{one-shot}}^*
&\triangleq&\inf\{L:
(\mathsf \pi,L)\textnormal{ is feasible}\},
\end{IEEEeqnarray}
where $\pi\in\Omega$ denotes a one-shot communication scheme.
\end{definition}

The objective of this paper is to jointly design the newly joined clusters' caches, Reduce function
reassignment, and Shuffle communication to lower the communication load $L$,
while preserving the surviving initial workers' caches and Reduce assignments.
\subsection{Preliminaries}
\label{sec:preliminaries}
We recall the PDA and Cartesian product definitions used below.
\begin{definition}[Placement Delivery Array, abbreviated as PDA {\cite{7973185}}]
\label{def:pda}
	For positive integers $K,F,Z,S$, an $F\times K$ array $P$ whose entries are either $\star$ or integers in $[S]$ is a $(K,F,Z,S)$ placement delivery array (PDA) if it satisfies the following conditions:
	\begin{itemize}[label=]
		\item \textbf{C1:} The symbol $\star$ appears $Z$ times in each column.
		\item \textbf{C2:} Each integer $s\in [S]$ occurs at least once in the array.
		\item \textbf{C3:} For any two different coordinates $(f_1,k_1),(f_2,k_2)\in[F]\times[K]$, if $P(f_1,k_1)=P(f_2,k_2)=s\in[S]$, then
			\begin{itemize}
			\item $f_1 \neq f_2$ and $k_1 \neq k_2$;
			\item $P(f_1,k_2)=P(f_2,k_1)=\star$.
			\end{itemize}
	\end{itemize}
\end{definition}
\begin{definition}[All-Star Row and All-Star Column {\cite{wang2023placement}}]
Given a $(K,F,Z,S)$ PDA $P$ and a symbol $s\in[S]$, a row $f\in[F]$
is an all-star row of $s$ if $P(f,m)=\star$ for every
$(j,m)\in[F]\times[K]$ satisfying $P(j,m)=s$. Similarly, a column
$k\in[K]$ is an all-star column of $s$ if $P(j,k)=\star$ for every
$(j,m)\in[F]\times[K]$ satisfying $P(j,m)=s$.
\end{definition}
\begin{definition}[Basic PDA {\cite{wang2023placement}}]
\label{def:basic_pda}
	A $(K,F,Z,S)$ PDA is a basic PDA with repetition factor $\lambda\in\mathbb N^+$ if it also satisfies the following conditions:
	\begin{itemize}[label=]
		\item \textbf{C4:} The integer $\lambda$ divides $F$, and the star pattern
        repeats every $F/\lambda$ rows: for $f\in[F/\lambda]$, $k\in[K]$,
        and $i\in[0:\lambda-1]$, $P(f+iF/\lambda,k)=\star$ if and only if
        $P(f,k)=\star$.
		\item \textbf{C5:} There exists a map $\theta:[S]\to[F/\lambda]$ satisfying
			\begin{itemize}
				\item For any $s\in[S]$, row $\theta(s)$ is an all-star row of $s$;
				\item Each row $f\in[F/\lambda]$ is assigned the same number of
                symbols: $|\{s\in[S]:\theta(s)=f\}|=\lambda S/F$.
			\end{itemize}
	\end{itemize}
\end{definition}
For $g\in\mathbb N^+$, a PDA is called $g$-regular if each of its
non-star symbols occurs exactly $g$ times.
\begin{definition}[Cartesian Product of Arrays {\cite{wang2023placement}}]
For positive integers $F_1,K_1,F_2,K_2$, an $F_1\times K_1$ array $P_1$,
and an $F_2\times K_2$ array $P_2$,
the Cartesian product $P_3=P_1\times P_2$ has $F_1F_2$ rows and $K_1+K_2$
columns. Its rows are indexed by $[F_1]\times[F_2]$, with the second
coordinate varying fastest: $(f_1,f_2)$ occupies row $(f_1-1)F_2+f_2$.
For each such pair,
\begin{IEEEeqnarray*}{rCl}
P_3[(f_1,f_2),k] &=& P_1[f_1,k],\qquad 1\leq k\leq K_1,\\
&=& P_2[f_2,k-K_1],\qquad K_1<k\leq K_1+K_2.
\end{IEEEeqnarray*}
\end{definition}
\section{Main Results}
\label{sec:results}
The proposed scheme, detailed in Section~\ref{sec:scheme}, uses
$\FCart=\prod_{c\in[C]}t_c\binom{K_c}{t_c}$ equal file batches, with
$\FCart$ denoting the number of rows in the Cartesian product array and
$N/\FCart\in\mathbb N^+$. Let $t_0=\sum_{c\in\mathcal C_{\rm ini}}t_c$
and $\tau=\sum_{c\in[C]}t_c$ be the initial and total replication levels.
Define
\begin{IEEEeqnarray*}{rCl}
\rho_c&=&\frac{N-M_c}{M_c}=\frac{K_c-t_c}{t_c},\qquad c\in[C],\\
\rho_{\max}&=&\max_{c\in\mathcal C_{\rm ini}}\rho_c,\qquad
\rho_{\min}=\min_{c\in\mathcal C_{\rm ini}}\rho_c.
\end{IEEEeqnarray*}
Thus $\rho_c$ is the ratio of uncached to cached files per worker;
its two extrema refer only to the initial clusters.

In our proposed scheme, the Shuffle phase consists of $U=\max_{k\in[K]}\gamma_k$ episodes.
The IVs for one Reduce function from one file batch form a \emph{packet}
of $NT/\FCart$ bits. In each episode, the requested packets are grouped
for delivery, with each group identified by a distinct label called a
\emph{delivery symbol} $s$. In episode $u\in[U]$, these symbols are partitioned into
$\mathcal S_1^u$, $\mathcal S_2^u$, and $\mathcal S_3^u$.
Each symbol in $\mathcal S_1^u\cup\mathcal S_2^u$ requires one packet,
whereas each $s\in\mathcal S_3^u$ is delivered cooperatively by its $n_s$
receivers using $n_s/(n_s-1)$ packets, where $n_s\in[2:\tau]$.
The grouping, partition, and
transmission rules are specified in
Section~\ref{sec:delivery}.
\begin{theorem}[Communication load]
\label{thm:load}
Given the file placement $\mathcal M$ and Reduce function assignment $\boldsymbol\gamma$, let $|\mathcal D\cap\mathcal K_c|$ denote the number of disconnected workers in cluster $c\in[C]$. Assume there exist at least one cluster $c\in [C]$ satisfies the $|\mathcal D\cap\mathcal K_c|<t_c$. Then the proposed scheme attains the achievable load
\begin{IEEEeqnarray}{rCl}
L_{\rm dyn}&=&\sum_{u=1}^{U}
\frac{|\mathcal S_1^u|+|\mathcal S_2^u|
+\sum_{s\in\mathcal S_3^u}\frac{n_s}{n_s-1}}
{K_0\FCart},
\end{IEEEeqnarray}
\end{theorem}
\begin{IEEEproof}
The achievability is proved by the proposed scheme in Section \ref{sec:scheme}.
\end{IEEEproof}

Corollary~\ref{pro2} specializes Theorem~\ref{thm:load} to the
no-disconnection setting.

\begin{samepage}
\begin{corollary}[No-disconnection load]
\label{pro2}
Assume that no worker disconnects. If a new cluster joins, or if $\tau=1$
(necessarily $C=C_{\rm ini}=1$ and $t_1=1$), then
\begin{IEEEeqnarray}{rCl}
L_{\rm dyn}&=&\frac{\rho_{\max}}{K_0}.
\label{load_new}
\end{IEEEeqnarray}
If $C_{\rm new}=0$ and $\tau>1$, then
\begin{IEEEeqnarray}{rCl}
\frac{\rho_{\max}}{K_0}
\leq L_{\rm dyn}\leq
\frac{1}{K_0}\left(\rho_{\max}+
\frac{\rho_{\min}}{\tau-1}\right).
\label{load_no_new}
\end{IEEEeqnarray}
\end{corollary}
\end{samepage}

\begin{IEEEproof}
The proof is given in Appendix~\ref{app:proof-no-disconnection}.
\end{IEEEproof}
When $C=C_{\rm ini}=1$, the load is $\rho_{\max}/K_0$ and meets
Lemma~\ref{lem:lower}. When $C=C_{\rm ini}\geq2$ and
$\rho_{\max}=\rho_{\min}$, the upper endpoint of \eqref{load_no_new} is attained.

\begin{lemma}[General lower bound \cite{li2017fundamental}]
\label{lem:lower}
Consider a MapReduce-based computing system with $N$ input files, $Q$ output functions and $K$ workers.
For $s_0\in[K]$ and $d\in[K-s_0]$, let $a_{s_0,d}$ denote the number of IVs
required by exactly $d$ connected workers that do not cache them and available
at exactly $s_0$ other connected workers. Every feasible CDC scheme satisfies the converse bound
\begin{IEEEeqnarray}{rCl}
	L^*&\geq&\frac{1}{QN}\sum_{s_0=1}^K\sum_{d=1}^{K-s_0}a_{s_0,d}\frac{d}{s_0+d-1}.\label{lower_bound}
\end{IEEEeqnarray}
\end{lemma}

The next theorem gives a converse bound for one-shot communication
under the same connectivity and Reduce function assignment as the proposed
scheme. For $\mathcal M\in\mathcal P$, $c\in[C]$, and $n\in[N]$, define
\begin{IEEEeqnarray}{rCl}
r_{c,n}&\triangleq&\bigl|\{i\in[K_c]:k_i^{(c)}\in\bar{\mathcal D},\
n\in\mathcal M_i^{(c)}\}\bigr|,\\
d_{c,n}&\triangleq&\sum_{\substack{i\in[K_c]:\ k_i^{(c)}\in\bar{\mathcal D},\\n\notin\mathcal M_i^{(c)}}}\gamma_{k_i^{(c)}}.
\end{IEEEeqnarray}
Here $r_{c,n}$ counts connected copies of file $n$ in cluster $c$, and
$d_{c,n}$ counts its IVs requested but not cached by their requesting
workers in that cluster.

\begin{theorem}[One-shot converse]
\label{thm:multicast}
For any feasible file placement $\mathcal M\in\mathcal P$ and Reduce function assignment $\boldsymbol{\gamma}\in\mathcal{Q}$, the minimum communication load in
$\Omega$ satisfies
\begin{IEEEeqnarray}{rCl}
L_{\textnormal{one-shot}}^*
&\geq&\frac{1}{NK_0}\max_{c\in[C]}
\sum_{n=1}^{N}\frac{d_{c,n}}{r_{c,n}+1},
\label{eq:placement_lb}
\end{IEEEeqnarray}
where $\mathcal P$ and $\Omega$ denote the feasible placement class and one-shot communication class, respectively.

If no worker disconnects, then
\begin{IEEEeqnarray}{rCl}
L_{\textnormal{one-shot}}^*
&\geq&\frac{1}{K_0}\max_{c\in\mathcal C_{\rm ini}}
\frac{K_c-t_c}{t_c+1}.
\label{eq:placement_free_lb}
\end{IEEEeqnarray}
\end{theorem}

\begin{IEEEproof}
The proof is given in Appendix~\ref{app:proof-multicast}.
\end{IEEEproof}

We now combine the achievable load in Corollary~\ref{pro2} with
Theorem~\ref{thm:multicast} to establish a constant approximation guarantee
under one-shot communication in the class $\Omega$.
\begin{theorem}[Constant-factor guarantee without disconnections]
\label{thm:gap}
Consider the file placement $\mathcal{M}$ and Reduce function assignment $\boldsymbol\gamma$ of Theorem~\ref{thm:load}. Assume that no worker disconnects and $\tau>1$, and consider the
one-shot communication class $\Omega$. Let $t_\star$ be
the largest replication level among initial clusters attaining
$\rho_{\max}$:
\begin{IEEEeqnarray}{rCl}
t_\star=\max\{t_c:c\in\mathcal C_{\rm ini},\ 
\rho_c=\rho_{\max}\}.
\end{IEEEeqnarray}
For the given placement $\mathcal M\in\mathcal P$ and Reduce function assignment $\boldsymbol\gamma$,
\begin{IEEEeqnarray}{l}
\displaystyle\frac{L_{\rm dyn}}{L_{\textnormal{one-shot}}^*(\mathcal M,\boldsymbol\gamma)}\leq\nonumber\\[2pt]
\quad\displaystyle\left(1+\frac{1}{t_\star}\right)
\begin{cases}
1,&C_{\rm new}\geq1,\\[5pt]
\displaystyle1+\frac{\rho_{\min}}{\rho_{\max}(t_0-1)},
&C_{\rm new}=0.
\end{cases}
\label{eq:placement_gap}
\end{IEEEeqnarray}
For any feasible file placement $\mathcal M\in\mathcal P$, we have
\begin{IEEEeqnarray}{rCl}
\frac{L_{\rm dyn}}
{\min_{\mathcal M\in\mathcal P}L_{\textnormal{one-shot}}^*(\mathcal M,\boldsymbol\gamma)}
&\leq&\begin{cases}2,&C_{\rm new}\geq1,\\[2pt]4,&C_{\rm new}=0.\end{cases}
\label{eq:unconditional_gap}
\end{IEEEeqnarray}
\end{theorem}
\begin{IEEEproof}
The proof is given in Appendix~\ref{app:proof-gap}.
\end{IEEEproof}

Within the class $\Omega$, these constants are
independent of cluster sizes, cache-ratio imbalance, and arriving replication.
For $\tau=1$, the scheme is optimal and \eqref{eq:unconditional_gap} also holds.
\section{Numerical Results}
\label{sec:numerical}
Since none of the existing CDC schemes jointly address file
placement for newly joined workers and reassignment of abandoned Reduce
functions, we extend decentralized coded caching~\cite{6807823}
to construct a decentralized CDC baseline for comparison.
We use $30$ random realizations per setting. Disconnected workers
are sampled uniformly from the initial clusters. The baseline uses its own
random placement and Reduce function reassignment.
All percentage reductions below are relative to the baseline communication load.

In the figures, $L_{\textnormal{dec}}$ and $L_{\textnormal{dyn}}$ denote the
communication loads achieved by the decentralized baseline and the proposed
dynamic CDC, respectively. The General lower bound and One-shot lower bound
are given in Lemma~\ref{lem:lower} and Theorem~\ref{thm:multicast}, respectively,
evaluated under the proposed placement and reassignment. The latter restricts
communication to the one-shot class $\Omega$.

\subsection{Decentralized CDC}
We adapt decentralized coded caching to CDC by treating the IVs required by
Reduce functions as demands. To ensure that every file is initially stored,
we first assign each file to a uniformly chosen worker with available cache
space. We repeat this pass, excluding workers that already cache the file,
until all caches are full. Each worker in cluster $c$ thus stores $M_c$
distinct files and computes their IVs. When workers are disconnected, abandoned Reduce
functions are randomly reassigned to distinct connected workers. The Shuffle
exploits random cache overlaps through coded multicast: workers transmit XORs of
locally available IV subpackets, and recipients decode using cached IVs as side
information. Thus, the workers themselves generate the coded messages, replacing
the server-based delivery in coded caching.

For the decentralized baseline, we count assignment and delivery computations,
excluding file transfer and storage I/O. Let $n=|\bar{\mathcal D}|$ and let $U_{\rm b}$ be the baseline's number
of episodes. File placement requires at most $K$ passes over the $N$ files,
each testing $K$ workers per file, giving $O(NK^2)$ operations.
Random function reassignment and indexing the surviving cache patterns take
$O(K+|\mathcal D_{\rm ini}|+Nn)$ operations.
Each episode enumerates
$E_{\rm b}=\sum_{s=2}^{t_{\max}+1}s\binom{n}{s}$
subset-worker pairs, where $t_{\max}$ is the largest number of connected
copies of a file not cached by every connected worker, or zero if no
such file exists. XOR encoding and decoding require
$O(NQT(t_{\max}+1))$ bit operations. The total computational overhead is
\begin{IEEEeqnarray}{l}
O\!\bigl(NK^2+U_{\rm b}E_{\rm b}+NQT(t_{\max}+1)\bigr).
\end{IEEEeqnarray}
\subsection{Comparisons}
\label{sec:arrival}
Fig.~\ref{fig:entry} first illustrates the effects of disconnected workers
and of the size of the arriving cluster.

\begin{figure}[t]
\centering
\captionsetup[subfloat]{farskip=0pt}
\includegraphics[width=\linewidth]{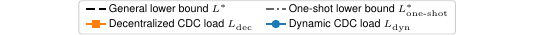}\\[0pt]
\subfloat[]{\includegraphics[width=0.49\linewidth]{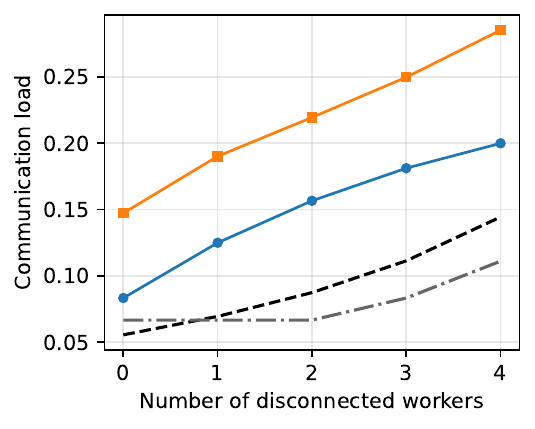}}
\hfill
\subfloat[]{\includegraphics[width=0.49\linewidth]{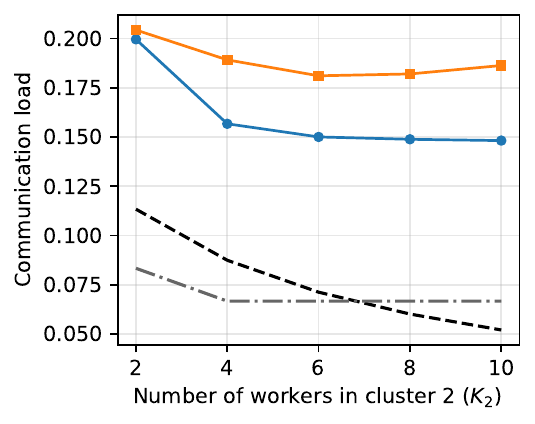}}
\caption{Worker-arrival experiments with
$C_{\textnormal{ini}}=C_{\textnormal{new}}=1$, $K_1=6$,
$M_1=2N/3$, $M_2=N/2$, and $Q=6$. (a) Communication load versus the
number of disconnected workers for $K_2=4$. (b) Communication
load versus $K_2\in\{2,4,6,8,10\}$ for two disconnected workers.}
\label{fig:entry}
\end{figure}

In Fig.~\ref{fig:entry}(a), the proposed scheme reduces the average
communication load by approximately $27.5\%$--$43.4\%$ relative to the
baseline across the tested numbers of disconnected workers.
The one-shot bound is higher with no disconnected workers, whereas the general
bound is higher at the other four points.
In Fig.~\ref{fig:entry}(b), the reduction is approximately $2.3\%$ at
$K_2=2$ and $17.1\%$--$20.5\%$ for $K_2\geq4$. These reductions
illustrate the growing benefit of structured cross-cluster multicast
opportunities as the arriving cluster becomes larger. In the same figure,
the one-shot bound is the higher converse at $K_2=8$ and $10$,
where the general bound continues to decrease.

We next fix $K_1=6$, $K_2=5$, $M_1=N/2$, and three disconnected workers, while
varying the cache size of each arriving worker in Fig.~\ref{fig:arrival_cache}.

\begin{figure}[t]
\centering
\includegraphics[width=0.7\linewidth]{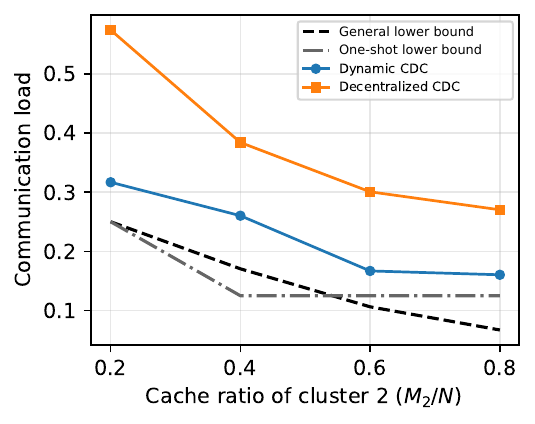}
\caption{Communication load versus the arriving-worker cache ratio
$M_2/N\in\{1/5,2/5,3/5,4/5\}$ for
$C_{\textnormal{ini}}=C_{\textnormal{new}}=1$, $K_1=6$, $K_2=5$,
$M_1=N/2$, three disconnected workers, and $Q=6$.}
\label{fig:arrival_cache}
\end{figure}

Both loads decrease as $M_2/N$ increases. The proposed scheme reduces the
average communication load by approximately $32.2\%$--$44.8\%$ relative
to the baseline across the four cache ratios. At $M_2/N=1/5$, worker
disconnections leave at least one file with no surviving copy in $6$ of the $30$
realizations, making the baseline infeasible. Its average load uses the
remaining $24$ realizations. At the other cache ratios, the baseline is
feasible in all $30$ realizations, which are all included in the averages. The
one-shot bound is higher at $M_2/N=3/5$ and $4/5$, where it is
$0.125$ at both points, compared with proposed loads $0.167$ and $0.160$,
respectively.

\begingroup
\setlength{\parskip}{0pt}
\section{Illustrative Example}
\label{sec:illustrative}
\begin{figure}[t]
\centering
\includegraphics[width=0.8\columnwidth]{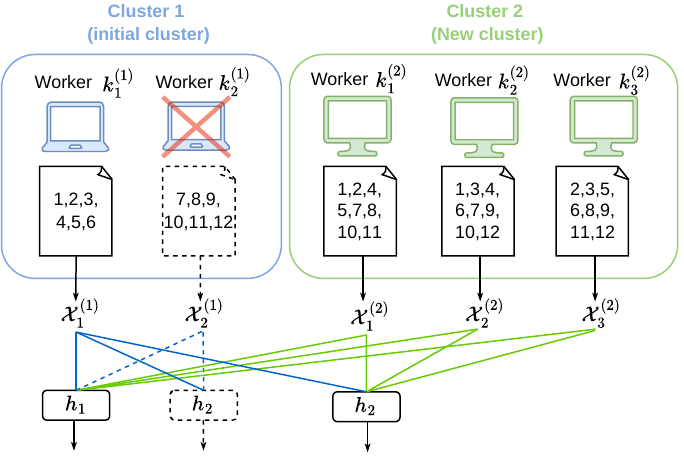}
\caption{The dynamic CDC system with $K_1=2$ initial workers. Worker $k^{(1)}_2$ disconnects, while $K_2=3$ new workers join the system.}
\label{fig:example}
\end{figure}
We now present an illustrative example that conveys the key idea of our
scheme; the corresponding dynamic CDC system appears in Fig.~\ref{fig:example}.
Consider a MapReduce computing system that starts with $K_1=2$ workers, each
of which aims to compute a distinct output function, and $N=12$ input files,
i.e., $Q=2$.
During operation, worker $k^{(1)}_2$ disconnects, and a new cluster of $K_2=3$
workers joins the system. The numbers of files cached by the workers in the two
clusters are $M_1=6$ and $M_2=8$, respectively.

First, construct the MN PDAs:
\begin{IEEEeqnarray*}{cc}
   P_{\textnormal{MN}}^{(1)} =\left[\begin{array}{cc}
        \star & 1 \\
        1 & \star
    \end{array}\right],\quad&
    P_{\textnormal{MN}}^{(2)} = \left[\begin{array}{ccc}
        \star & \star &1 \\
        \star & 1 & \star \\
        1 & \star & \star
    \end{array}\right].
\end{IEEEeqnarray*}
The corresponding basic PDAs are as follows:
\begin{IEEEeqnarray*}{cc}
    P_{\textnormal{basic}}^{(1)}=\left[
    \begin{array}{cc}
    \star & 1 \\
      2   & \star
    \end{array}
    \right],&\quad
    P_{\textnormal{basic}}^{(2)}=\left[
    \begin{array}{ccc}
        \star & \star & 1  \\
        \star & 2 & \star \\
        3 & \star & \star  \\
        \star & \star & 3  \\
        \star & 1 & \star  \\
        2 & \star & \star  
    \end{array}
    \right].
\end{IEEEeqnarray*}
The basic PDAs are combined through their Cartesian product to enable
coding gains across clusters.
Section~\ref{sec:entrywise} details the construction of the Communication Array.
Forming the Cartesian product of $P_{\textnormal{basic}}^{(1)}$ and
$P_{\textnormal{basic}}^{(2)}$ then gives
\par
\begin{IEEEeqnarray}[\footnotesize\renewcommand{\arraystretch}{0.92}
\abovedisplayskip=6pt\belowdisplayskip=6pt
\abovedisplayshortskip=6pt\belowdisplayshortskip=6pt]{c}
	\label{eq:illustrative_cart}
	P'=P_{\textnormal{basic}}^{(1)}\times P_{\textnormal{basic}}^{(2)} =\left[
	 \begin{array}{cc:ccc}
        \star & 1 & \star & \star & 1\\
        \star & 1 & \star & 2 & \star\\
        \star & 1 & 3 & \star & \star\\
        \star & 1 & \star & \star & 3\\
        \star & 1 & \star & 1 & \star\\
        \star & 1 & 2 & \star & \star\\
        2 & \star & \star & \star & 1\\
        2 & \star & \star & 2 & \star\\
        2 & \star & 3 & \star & \star\\
        2 & \star & \star & \star & 3\\
        2 & \star & \star & 1 & \star\\
        2 & \star & 2 & \star & \star
    \end{array}\right].
\end{IEEEeqnarray}

Each row of $P'$ represents a file subset (one file in this example),
and each column corresponds to a worker. The columns are ordered as
$(k_1^{(1)},k_2^{(1)},k_1^{(2)},k_2^{(2)},k_3^{(2)})$.
Each worker caches the files corresponding to the star elements' rows in its column,
giving the following placement.
\begin{IEEEeqnarray*}[\abovedisplayskip=6pt\belowdisplayskip=6pt
\abovedisplayshortskip=6pt\belowdisplayshortskip=6pt]{rCl}
	\mathcal{M}^{(1)}_1&=&\{1,2,3,4,5,6\},\\
	\mathcal{M}^{(1)}_2&=&\{7,8,9,10,11,12\},\\
	\mathcal{M}^{(2)}_1&=&\{1,2,4,5,7,8,10,11\},\\
	\mathcal{M}^{(2)}_2&=&\{1,3,4,6,7,9,10,12\},\\
	\mathcal{M}^{(2)}_3&=&\{2,3,5,6,8,9,11,12\}.
\end{IEEEeqnarray*}

We compare reassigning the abandoned function to $k_1^{(1)}$ or, by symmetry
within cluster $2$, to $k_1^{(2)}$. The candidate Reduce function assignment
vectors are $(2,0,0,0,0)$ and $(1,0,1,0,0)$, respectively, with entries ordered as
$(k_1^{(1)},k_2^{(1)},k_1^{(2)},k_2^{(2)},k_3^{(2)})$.

To handle multiple assigned functions, we divide Shuffle into \emph{episodes},
each serving at most one Reduce function per worker. Each episode is further
divided into adjustment \emph{rounds}; the total round count across all episodes
is used as the simulated communication load. For
$\boldsymbol\gamma=(1,0,1,0,0)$, one episode suffices.
Let $\mathcal R=\{1,2\}\times\{1,2,3\}$, with tuples ordered as
$(r_1,r_2)=(1,1),(1,2),(1,3),(2,1),(2,2),(2,3)$.
Using \eqref{number_elem}, the round counts for these tuples are
$(2,2,2,1,1,0)$, giving simulated communication load $L_{\textnormal{sim}}=8$.

For $\boldsymbol\gamma=(2,0,0,0,0)$, each of the two episodes has assignment
vector $(1,0,0,0,0)$ and round counts $(2,2,2,0,0,0)$, giving total simulated
load $2(2+2+2)=12>8$. We therefore assign the abandoned function to $k_1^{(2)}$.
Following Section~\ref{sec:comPDA}, we adjust the non-star entries of $P'$
based on worker disconnections and $\boldsymbol\gamma$ so that delivery
symbols satisfy the cross-star condition C3, yielding the Communication Array
$P_{\textnormal{com}}$ in \eqref{P_com}. Entries not assigned a delivery
label are marked $\circ$ and generate no transmission;
Section~\ref{sec:entrywise} details the construction.
\par
\begingroup
\footnotesize
\setlength{\arraycolsep}{3pt}
\renewcommand{\arraystretch}{0.92}
\begin{IEEEeqnarray}[\abovedisplayskip=6pt\belowdisplayskip=6pt
\abovedisplayshortskip=6pt\belowdisplayshortskip=6pt]{c}
    \label{P_com}
    P_{\textnormal{com}} = \left[
    \begin{array}{cc:ccc}
        k_1^{(1)} & k_2^{(1)} & k_1^{(2)} & k_2^{(2)} & k_3^{(2)}\\
        \hdashline
        \star & \circ & \star & \star & \circ\\
        \star & \circ & \star & \langle2,2\rangle_{2} & \star\\
        \star & \circ & \langle2,3\rangle_{4} & \star & \star\\
        \star & \circ & \star & \star & \langle2,3\rangle_{4}\\
        \star & \circ & \star & \circ & \star\\
        \star & \circ & \langle2,2\rangle_{2} & \star & \star\\
        \langle2,2\rangle_{2} & \star & \star & \star & \circ\\
        \langle2,3\rangle_{4} & \star & \star & \langle2\rangle_{7} & \star\\
        \langle2\rangle_{6} & \star & \langle3\rangle_{8} & \star & \star\\
        \langle2\rangle_{1} & \star & \star & \star & \langle3\rangle_{8}\\
        \langle2\rangle_{3} & \star & \star & \circ & \star\\
        \langle2\rangle_{5} & \star & \langle2\rangle_{7} & \star & \star
    \end{array}
    \right]    
\end{IEEEeqnarray}
\endgroup

During the Shuffle phase, we process each non-null delivery symbol in the
Communication Array. Consider, for instance, the element $\langle2,2\rangle_{2}$, which
occurs at row-column positions $(7,1)$, $(2,4)$, and $(6,3)$ in the displayed array. Extracting the subarray of
$P_{\textnormal{com}}$ at rows $\{2,6,7\}$ and columns $\{1,3,4\}$ yields
\begin{IEEEeqnarray*}{c}
\left[
    \begin{array}{c:ccc}
    &  k^{(1)}_1 & k^{(2)}_1 & k^{(2)}_2\\
    \hdashline
        2 &  \star & \star & \langle2,2\rangle_{2}\\
        6 &  \star & \langle2,2\rangle_{2} & \star\\
        7 &  \langle2,2\rangle_{2} & \star & \star
    \end{array}\right],
\end{IEEEeqnarray*}
where only $k^{(1)}_1$ and $k^{(2)}_1$ are assigned Reduce functions, so worker
$k^{(2)}_2$, although it is assigned no Reduce function, acts as the transmitter and
multicasts $v_{1,7}\oplus v_{2,6}$ to workers $k^{(1)}_1$ and $k^{(2)}_1$.

For the element $\langle2\rangle_{1}$ occurring at $(10,1)$, the connected all-star columns are
$\{3,4\}$, i.e., $P_{\textnormal{com}}[10,3]=\star$ and
$P_{\textnormal{com}}[10,4]=\star$, and we take one of them, say column $3$, as
the transmitter. The subarray at row $10$ and columns $\{1,3\}$ reads
\begin{IEEEeqnarray*}{c}
\left[
    \begin{array}{c:cc}
    & k^{(1)}_1 & k^{(2)}_1\\
    \hdashline
        10 & \langle2\rangle_{1} & \star
    \end{array}\right].
\end{IEEEeqnarray*}
Worker $k^{(2)}_1$, which corresponds to this all-star column, transmits the IV
$v_{1,10}$ to worker $k^{(1)}_1$.

Considering the non-null delivery symbols one by one, we complete the transmission
procedure. The overall communication is as follows. Worker $k^{(2)}_1$
multicasts $v_{1,10}$ and $v_{1,11}$ to worker $k^{(1)}_1$; worker $k^{(2)}_2$
multicasts $v_{1,7}\oplus v_{2,6}$ to workers $k^{(1)}_1$ and $k^{(2)}_1$,
$v_{1,9}$ and $v_{1,12}$ to worker $k^{(1)}_1$, and $v_{2,12}$ to worker
$k^{(2)}_1$; worker $k^{(2)}_3$ multicasts $v_{1,8}\oplus v_{2,3}$ to workers
$k^{(1)}_1$ and $k^{(2)}_1$, and $v_{2,9}$ to worker $k^{(2)}_1$.
The communication load is $L=\frac{8}{12\cdot2}=\frac{1}{3}$.
 
The alternative assignment to $k^{(1)}_1$ yields load
$\frac{12}{12\cdot2}=\frac{1}{2}>\frac{1}{3}$.
The simulated communication load enables efficient selection of a suitable
Reduce function assignment vector without computing the actual communication
load for every possible assignment.
\par\endgroup
\section{General Achievable Scheme}
\label{sec:scheme}
The scheme uses a Cartesian product of cluster-wise basic PDAs for file placement
and adjusts it after reassignment to generate coded Shuffle transmissions.
Before Shuffle, $C_{\textnormal{new}}$ new clusters join the
$C_{\textnormal{ini}}$ initial clusters, giving
$C=C_{\textnormal{ini}}+C_{\textnormal{new}}$ clusters in total.
We first describe the construction of the Communication Array for a given
Reduce function assignment vector. We then use its round counts to choose
that vector during reassignment, before constructing the Communication Arrays and performing Shuffle.
\subsection{File Placement}
\label{sec:placement}
We use the PDA representation \cite{7973185} of the MN
scheme \cite{6763007}.

\subsubsection{Construction of MN PDAs} For each cluster $c\in[C]$, construct a
$(t_c+1)$-regular
$\bigl(K_c,\binom{K_c}{t_c},\binom{K_c-1}{t_c-1},
\binom{K_c}{t_c+1}\bigr)$ MN PDA $P_{\textnormal{MN}}^{(c)}$ as follows.

Enumerate the $(t_c+1)$-subsets of $[K_c]$ as
$\mathcal{S}^{(c)}_1,\ldots,\mathcal{S}^{(c)}_{\binom{K_c}{t_c+1}}$, and
enumerate the $t_c$-subsets as
$\mathcal T^{(c)}_1,\ldots,\mathcal T^{(c)}_{\binom{K_c}{t_c}}$ to index
the rows of $P^{(c)}_{\textnormal{MN}}$. For row $f\in[\binom{K_c}{t_c}]$
and column $k\in[K_c]$,
\begin{IEEEeqnarray}{rCl}
    P_{\textnormal{MN}}^{(c)}[f,k]=\left\{
        \begin{array}{ll}
            \star, & k\in \mathcal T_f^{(c)}, \\
            s, & k \notin \mathcal T_f^{(c)},\ \mathcal T_f^{(c)}\cup\{k\}=\mathcal S_s^{(c)}
        \end{array}.
    \right.
\end{IEEEeqnarray}
\subsubsection{Construction of Basic PDAs from MN PDAs}
For each $P_{\textnormal{MN}}^{(c)}$, $c\in[C]$, the construction of
\cite{wang2023placement} yields a $t_c$-regular
$\bigl(K_c,t_c\binom{K_c}{t_c},t_c\binom{K_c-1}{t_c-1},
(t_c+1)\binom{K_c}{t_c+1}\bigr)$ basic PDA
$P_{\textnormal{basic}}^{(c)}$ with repetition factor $t_c$.
Stack $t_c$ copies of $P^{(c)}_{\textnormal{MN}}$, indexed by
$h\in[0:t_c-1]$, preserving all stars. For each MN symbol $s$, number its
$t_c+1$ occurrences by increasing row index as $j=0,\ldots,t_c$.
In copy $h$, replace occurrence $j$ by
$(t_c+1)(s-1)+1+((j-h)\bmod(t_c+1))$.
Thus each new symbol occurs $t_c$ times in the stacked array.

\subsubsection{Cartesian Product}
Form the Cartesian product of the $C$ basic PDAs in increasing cluster order:
$P_{\textnormal{cart}}=P_{\textnormal{basic}}^{(1)}\times\cdots\times P_{\textnormal{basic}}^{(C)}$.
The array has \mbox{$\FCart=\prod_{c\in[C]}t_c\binom{K_c}{t_c}$} rows
and $K$ columns. Divide the $N$ files, in their original order, into
$\FCart$ equal, disjoint batches, one per row,
with $N/\FCart\in\mathbb N^+$.
For $c\in[C]$, $i\in[K_c]$, and $b\in[\FCart]$, worker $k_i^{(c)}$ caches all files in batch $b$ if and only if
	\begin{IEEEeqnarray}{rCl}
		P_{\textnormal{cart}}\Bigl[b,\; \sum_{a=1}^{c-1} K_a + i\Bigr] = \star.
	\end{IEEEeqnarray}
The column position $\sum_{a=1}^{c-1}K_a+i$ equals the worker's global
index $k_i^{(c)}$.
\begin{remark}
File placement can proceed without waiting for all clusters to join.
Initially, apply the above placement rule to the Cartesian product of
the initial clusters' basic PDAs. When cluster $i\in\mathcal C_{\rm new}$
arrives, append its basic PDA to the right of the current array:
$P_{\textnormal{cart}}\leftarrow P_{\textnormal{cart}}\times P_{\textnormal{basic}}^{(i)}$.
Split each previous batch into $t_i\binom{K_i}{t_i}$ consecutive
subbatches in the same file order, and assign them to the new workers
according to the stars in their columns. Thus each new cluster receives
its file allocation upon arrival, while the file indices and the files
cached by existing workers remain unchanged.
\end{remark}

Write $t_0=\sum_{c\in\mathcal C_{\rm ini}}t_c$ and
$\tau=\sum_{c\in[C]}t_c$ for the initial and total replication levels,
respectively. Every row of a cluster-$c$ basic PDA has exactly $t_c$ stars,
so every file has exactly $t_c$ copies in cluster $c$ and $\tau$ copies
overall before disconnections.

\begin{proposition}[Disconnection tolerance]
\label{prop:feasibility}
For the Cartesian product placement, let
$\delta_c=|\mathcal D\cap\mathcal K_c|$, $c\in[C]$, and let $R_n$ denote
the number of connected workers caching file $n$. Then
\begin{IEEEeqnarray}{rCl}
\min_{n\in[N]}R_n
&=&\sum_{c\in[C]}\max\{t_c-\delta_c,0\}.
\label{eq:min_surviving_copies}
\end{IEEEeqnarray}
Every file remains available, and hence the proposed scheme can complete the
computation, if and only if $\delta_c<t_c$ for at least one cluster $c$.
In particular, the scheme tolerates every disconnection set satisfying
\begin{IEEEeqnarray}{rCl}
|\mathcal D|&<&\tau
=\sum_{c\in[C]}\frac{K_cM_c}{N}.
\label{eq:disconnection_guarantee}
\end{IEEEeqnarray}
This worst-case guarantee is tight: there exists a disconnection set of
size $\tau$ for which a file has no surviving copy.
\end{proposition}
\begin{IEEEproof}
The proof is given in Appendix~\ref{app:proof-feasibility}.
\end{IEEEproof}

The condition $|\mathcal D|<\tau$ guarantees feasibility regardless of which
clusters lose workers. When $|\mathcal D|\geq\tau$, feasibility still holds
if some cluster satisfies $\delta_c<t_c$.

\subsection{Construction of the Communication Array}
\label{sec:comPDA}
The Cartesian product $P_{\textnormal{cart}}$ guides file placement.
We now adjust its non-star entries to obtain Communication Arrays that also
specify the Shuffle transmissions. The stars remain unchanged.

\begin{definition}[Row and column index functions]
For a non-star symbol $s$ of $P_{\textnormal{basic}}^{(i)}$, define
\begin{IEEEeqnarray}{rCl}
\operatorname{Row}_i(s)
&=&\{f:\exists\ell,\ P_{\textnormal{basic}}^{(i)}[f,\ell]=s\},\\
\operatorname{Col}_i(s)
&=&\{\ell:\exists f,\ P_{\textnormal{basic}}^{(i)}[f,\ell]=s\}.
\end{IEEEeqnarray}
These are its row and column index sets in the basic PDA.
\end{definition}

Given the Reduce function assignment vector $\boldsymbol\gamma$, order the
functions at each worker $k$ by increasing index as
$\mathcal Q_k=\{q_{k,1},\ldots,q_{k,\gamma_k}\}$.
We construct $U=\max_{k\in[K]}\gamma_k$ Communication Arrays
$P_{\textnormal{com}}^1,\ldots,P_{\textnormal{com}}^U$, one for each of
$U$ episodes. In episode $u\in[U]$, worker $k$ computes $\phi_{q_{k,u}}$ if
$u\leq\gamma_k$ and otherwise has no demand. Thus the episode-wise
Reduce function assignment vector is
$\boldsymbol\gamma^u=(\gamma_1^u,\ldots,\gamma_K^u)$, where
$\gamma_k^u=\mathbf1\{u\leq\gamma_k\}$.
Define the connected workers with a function in this episode and their
clusters by
\begin{IEEEeqnarray}{rCl}
\mathcal K_{\rm act}^u
&=&\{k\in\bar{\mathcal D}:\gamma_k^u=1\},\\
\mathcal C_{\rm act}^u
&=&\{c\in[C]:\mathcal K_c\cap\mathcal K_{\rm act}^u\ne\emptyset\}.
\end{IEEEeqnarray}

Fix an episode $u$. Given $P_{\textnormal{cart}}$, $\boldsymbol\gamma^u$,
and $\mathcal D$, construct $P_{\textnormal{com}}^u$ as follows.
All episodes share $P_{\textnormal{cart}}$ and the occurrence indices.
Record only the new labels and their positions, using episode tags to
identify current labels. Positions without a current label are read as
$\star$ or $\circ$ according to the star pattern of $P_{\textnormal{cart}}$.
The construction numbers new labels consecutively by $g$, initially $g=1$.
Only activity states, counters, and pointers are reset for each episode;
the dependence on $u$ is omitted below.

Write $B_i=\binom{K_i}{t_i}$ for $i\in[C]$. The rows of
$P_{\textnormal{cart}}$ have coordinates
$\boldsymbol f=(f_1,\ldots,f_C)\in[t_1B_1]\times\cdots\times[t_CB_C]$.
Define column block $i$ as the submatrix $P^{(i)}$ with
\begin{IEEEeqnarray}{rCl}
P^{(i)}[\boldsymbol f,\ell]
&\triangleq&P_{\textnormal{cart}}\Bigl[\boldsymbol f,\sum_{a=1}^{i-1} K_a+\ell\Bigr] \nonumber \\
&=&P_{\textnormal{basic}}^{(i)}[f_i,\ell],\qquad \ell\in[K_i].
\end{IEEEeqnarray}
Thus $P_{\textnormal{cart}}=[P^{(1)},\ldots,P^{(C)}]$, and block $i$
contains the columns for cluster $i$.

We process every tuple in
\begin{IEEEeqnarray}{rCl}
\mathcal R&=&[B_1]\times\cdots\times[B_C].
\end{IEEEeqnarray}
For a fixed tuple $\boldsymbol r=(r_1,\ldots,r_C)\in\mathcal R$,
define the row indices in $P_{\textnormal{basic}}^{(c)}$ with the same
star positions as row $r_c$ by
\begin{IEEEeqnarray}{rCl}
\mathcal Z_c&=&\{r_c+hB_c:h\in[0:t_c-1]\},\qquad c\in[C].
\end{IEEEeqnarray}

For each non-star symbol $s$ of $P_{\textnormal{basic}}^{(i)}$, define
the Cartesian product row set
\begin{IEEEeqnarray}{l}
\mathcal F_{i,\boldsymbol r}(s)
\triangleq\mathcal Z_1\times\cdots\times\mathcal Z_{i-1}
\times\operatorname{Row}_i(s)\nonumber\\
\qquad{}\times\mathcal Z_{i+1}\times\cdots\times\mathcal Z_C.
\label{eq:compatible_group_coordinates}
\end{IEEEeqnarray}
Write $\boldsymbol f_{-i}$ for $\boldsymbol f$ with its $i$-th coordinate
removed. Define the candidate copies in block $i$ by
\begin{IEEEeqnarray}{rCl}
\mathcal H_{i,\boldsymbol r}
&\triangleq&\bigl\{(s,\boldsymbol f_{-i}):
\boldsymbol f\in\mathcal F_{i,\boldsymbol r}(s), \nonumber\\
&&\quad r_i\text{ is an all-star row of }s\bigr\}.
\end{IEEEeqnarray}
Each pair $(s,\boldsymbol f_{-i})$ identifies one \emph{copy of $s$},
consisting of $t_i$ entries in block $i$ that are replaced together.

Each eligible symbol has occurrence columns
$\operatorname{Col}_i(s)=\mathcal T_{r_i}^{(i)}$, since both sets have
$t_i$ members and row $r_i$ is an all-star row of $s$.
For any $f\in[B_i]$, define
\begin{IEEEeqnarray}{rCl}
\mathcal K_i^\star(f)
&\triangleq&\{k_\ell^{(i)}:\ell\in[K_i],\ P_{\textnormal{basic}}^{(i)}[f,\ell]=\star\},
\end{IEEEeqnarray}
the global indices of workers whose basic-PDA columns contain a star in
row $f$. In particular, $\mathcal K_i^\star(r_i)$ gives the workers in
the occurrence columns above.
If $\mathcal K_i^\star(r_i)\cap\mathcal K_{\rm act}^u=\emptyset$, set
$\mathcal H_{i,\boldsymbol r}\leftarrow\emptyset$, because none of these
copies has an active receiver. Otherwise, retain each copy in full,
including entries in inactive columns.
A symbol may have multiple copies. To specify which copy is selected
in each round, arrange the pairs in $\mathcal H_{i,\boldsymbol r}$ by
increasing symbol and choose a fixed order for copies of the same symbol.

It is easy to verify that $P_{\textnormal{basic}}^{(i)}$ has $K_i-t_i$
symbols for which row $r_i$ is an all-star row.
For each $\boldsymbol r\in\mathcal R$, we consider
$\prod_{c\ne i}t_c$ copies of each eligible symbol and define the number
of elements to be processed in block $i$ as
\begin{IEEEeqnarray}{rCl}
O_i^{(\boldsymbol r)}
&=&|\mathcal H_{i,\boldsymbol r}|\nonumber\\
&=&\begin{cases}
V_i,&\mathcal K_i^\star(r_i)\cap\mathcal K_{\rm act}^u\ne\emptyset,\\
0,&\text{otherwise},
\end{cases}
\label{number_elem}
\end{IEEEeqnarray}
where $V_i=(K_i-t_i)\prod_{c\ne i}t_c$.
Different copies of the same element are counted separately.
Let $p_i$ be the index of the next unprocessed pair in
$\mathcal H_{i,\boldsymbol r}$, initially $p_i=1$. Thus
$O_i^{(\boldsymbol r)}-(p_i-1)$ pairs remain to be processed.
In each round, we select one element from each chosen block and merge
the selected elements into a new label. 

The adjustment for $\boldsymbol r$ takes
$O_{\max}^{(\boldsymbol r)}=\max_{i\in[C]}O_i^{(\boldsymbol r)}$ rounds.
In each of these rounds, perform the following two steps.
\begin{enumerate}
\item \textit{Merge:} For each $i\in\mathcal C_{\rm act}^u$, block $i$
has $O_i^{(\boldsymbol r)}-p_i+1$ elements left to adjust.
If all these blocks have the same positive number of remaining elements,
select them all. Otherwise, select up to $|\mathcal C_{\rm act}^u|-1$
blocks with the most elements left to adjust, excluding blocks with no
remaining elements. If fewer such blocks remain, select all of them.
When several blocks have the same number of remaining elements,
select those with smaller cluster indices first.
Let $\mathcal N=\{i_1,\ldots,i_m\}$, with $i_1<\cdots<i_m$, be the
selected blocks. For each $i\in\mathcal N$, read the next pair
$(s_i,\boldsymbol z_{-i}^{(i)})=\mathcal H_{i,\boldsymbol r}[p_i]$.
Select its element $s_i$ for merging; the coordinates
$\boldsymbol z_{-i}^{(i)}$ identify the copy to be adjusted.
Merge the selected elements into the new label
\begin{IEEEeqnarray}{rCl}
\sigma_g&=&\langle s_{i_1},\ldots,s_{i_m}\rangle_g.
\label{eq:new_delivery_symbol}
\end{IEEEeqnarray}
This \emph{merged element} has $m$ components and labels
$\sum_{i\in\mathcal N}t_i$ entries. The round subscript distinguishes
labels with identical component tuples.

\item \textit{Substitute:} For each selected block $i\in\mathcal N$,
consider the rows $\boldsymbol f\in\mathcal F_{i,\boldsymbol r}(s_i)$
with $\boldsymbol f_{-i}=\boldsymbol z_{-i}^{(i)}$.
The entries equal to $s_i$ in $P^{(i)}$ at these rows form the copy
specified by the selected pair. Replace these $t_i$ entries by
$\sigma_g$ at the corresponding positions in $P_{\textnormal{com}}^u$.
For each such entry at local column $\ell$, set
\begin{IEEEeqnarray}{rCl}
P_{\textnormal{com}}^u
\Bigl[\boldsymbol f,\sum_{a=1}^{i-1} K_a+\ell\Bigr]
&\leftarrow&\sigma_g.
\label{eq:delivery_symbol_replacement}
\end{IEEEeqnarray}
Here $\ell$ is the local column index of the replaced entry.
After all replacements in this round, advance $p_i$ by one for the
selected clusters only, and advance $g$ by one.
\end{enumerate}
For a cluster not selected in a round, $p_i$ remains unchanged. Each
round reduces the largest remaining number by one, so all
lists are exhausted in exactly $O_{\max}^{(\boldsymbol r)}$ rounds.
If all lists are empty initially, no round is needed.
Repeat for every $\boldsymbol r\in\mathcal R$, continuing the label
index $g$ across tuples and retaining all labels already assigned in the
episode. Entries left as $\circ$ generate no transmission.

\begin{proposition}[Cross-star property and coverage]
\label{prop:communication_pda}
Every non-null symbol of $P_{\textnormal{com}}^u$ satisfies the cross-star condition C3
of Definition~\ref{def:pda}, and every
noncached batch requested by an active worker in episode $u$ is represented by
exactly one such symbol.
\end{proposition}
\begin{IEEEproof}
The proof is given in Appendix~\ref{app:proof-communication-pda}.
\end{IEEEproof}

\subsubsection{From the Cartesian product to the illustrative Communication Array}
\label{sec:entrywise}
Figure~\ref{fig:pda_steps} shows how the construction transforms $P'$
in \eqref{eq:illustrative_cart} into $P_{\textnormal{com}}$ in \eqref{P_com}.

\begingroup
\newcommand{\PDAcell}[1]{\makebox[3.4em]{$\textstyle #1$}}
\definecolor{PDAgroup1}{HTML}{0072B2}
\definecolor{PDAgroup2}{HTML}{D55E00}
\definecolor{PDAgroup3}{HTML}{E69F00}
\definecolor{PDAgroup4}{HTML}{009E73}
\definecolor{PDAgroup5}{HTML}{A65628}
\definecolor{PDAgroup6}{HTML}{7B2CBF}
\definecolor{PDAgroup7}{HTML}{CC3377}
\definecolor{PDAgroup8}{HTML}{008B9A}
\newcommand{\PDAchanged}[2]{{\color{PDAgroup#1}\boxed{\color{black}\textstyle #2}}}
\newcommand{\PDApending}[1]{{\color{black!48}#1}}

\begin{figure*}[t]
\centering
\scriptsize
\setlength{\arraycolsep}{0.6pt}
\setlength{\fboxsep}{1pt}
\renewcommand{\arraystretch}{1.0}
\begin{minipage}[t]{0.24\textwidth}
\centering
\parbox[t][2.6\baselineskip][t]{\linewidth}{\centering (a) Original Cartesian product $P'$}
\resizebox{\linewidth}{!}{$
\left[\begin{array}{r|cc:ccc}
f & \PDAcell{k_1^{(1)}} & \PDAcell{k_2^{(1)}} & \PDAcell{k_1^{(2)}} &
\PDAcell{k_2^{(2)}} & \PDAcell{k_3^{(2)}}\\[2pt]
\hline
1\vphantom{\boxed{8}} & \PDAcell{\star} & \PDAcell{\PDApending{1}} & \PDAcell{\star} & \PDAcell{\star} & \PDAcell{\PDApending{1}} \\
2\vphantom{\boxed{8}} & \PDAcell{\star} & \PDAcell{\PDApending{1}} & \PDAcell{\star} & \PDAcell{\PDApending{2}} & \PDAcell{\star} \\
3\vphantom{\boxed{8}} & \PDAcell{\star} & \PDAcell{\PDApending{1}} & \PDAcell{\PDApending{3}} & \PDAcell{\star} & \PDAcell{\star} \\
4\vphantom{\boxed{8}} & \PDAcell{\star} & \PDAcell{\PDApending{1}} & \PDAcell{\star} & \PDAcell{\star} & \PDAcell{\PDApending{3}} \\
5\vphantom{\boxed{8}} & \PDAcell{\star} & \PDAcell{\PDApending{1}} & \PDAcell{\star} & \PDAcell{\PDApending{1}} & \PDAcell{\star} \\
6\vphantom{\boxed{8}} & \PDAcell{\star} & \PDAcell{\PDApending{1}} & \PDAcell{\PDApending{2}} & \PDAcell{\star} & \PDAcell{\star} \\
7\vphantom{\boxed{8}} & \PDAcell{\PDApending{2}} & \PDAcell{\star} & \PDAcell{\star} & \PDAcell{\star} & \PDAcell{\PDApending{1}} \\
8\vphantom{\boxed{8}} & \PDAcell{\PDApending{2}} & \PDAcell{\star} & \PDAcell{\star} & \PDAcell{\PDApending{2}} & \PDAcell{\star} \\
9\vphantom{\boxed{8}} & \PDAcell{\PDApending{2}} & \PDAcell{\star} & \PDAcell{\PDApending{3}} & \PDAcell{\star} & \PDAcell{\star} \\
10\vphantom{\boxed{8}} & \PDAcell{\PDApending{2}} & \PDAcell{\star} & \PDAcell{\star} & \PDAcell{\star} & \PDAcell{\PDApending{3}} \\
11\vphantom{\boxed{8}} & \PDAcell{\PDApending{2}} & \PDAcell{\star} & \PDAcell{\star} & \PDAcell{\PDApending{1}} & \PDAcell{\star} \\
12\vphantom{\boxed{8}} & \PDAcell{\PDApending{2}} & \PDAcell{\star} & \PDAcell{\PDApending{2}} & \PDAcell{\star} & \PDAcell{\star} \\
\end{array}\right]$}
\end{minipage}
\renewcommand{\PDApending}[1]{{\color{black!48}\circ}}
\hfill
\begin{minipage}[t]{0.24\textwidth}
\centering
\parbox[t][2.6\baselineskip][t]{\linewidth}{\centering (b) $d=1$: round 1}
\resizebox{\linewidth}{!}{$
\left[\begin{array}{r|cc:ccc}
f & \PDAcell{k_1^{(1)}} & \PDAcell{k_2^{(1)}} & \PDAcell{k_1^{(2)}} &
\PDAcell{k_2^{(2)}} & \PDAcell{k_3^{(2)}}\\[2pt]
\hline
1\vphantom{\boxed{8}} & \PDAcell{\star} & \PDAcell{\PDApending{1}} & \PDAcell{\star} & \PDAcell{\star} & \PDAcell{\PDApending{1}} \\
2\vphantom{\boxed{8}} & \PDAcell{\star} & \PDAcell{\PDApending{1}} & \PDAcell{\star} & \PDAcell{\PDApending{2}} & \PDAcell{\star} \\
3\vphantom{\boxed{8}} & \PDAcell{\star} & \PDAcell{\PDApending{1}} & \PDAcell{\PDApending{3}} & \PDAcell{\star} & \PDAcell{\star} \\
4\vphantom{\boxed{8}} & \PDAcell{\star} & \PDAcell{\PDApending{1}} & \PDAcell{\star} & \PDAcell{\star} & \PDAcell{\PDApending{3}} \\
5\vphantom{\boxed{8}} & \PDAcell{\star} & \PDAcell{\PDApending{1}} & \PDAcell{\star} & \PDAcell{\PDApending{1}} & \PDAcell{\star} \\
6\vphantom{\boxed{8}} & \PDAcell{\star} & \PDAcell{\PDApending{1}} & \PDAcell{\PDApending{2}} & \PDAcell{\star} & \PDAcell{\star} \\
7\vphantom{\boxed{8}} & \PDAcell{\PDApending{2}} & \PDAcell{\star} & \PDAcell{\star} & \PDAcell{\star} & \PDAcell{\PDApending{1}} \\
8\vphantom{\boxed{8}} & \PDAcell{\PDApending{2}} & \PDAcell{\star} & \PDAcell{\star} & \PDAcell{\PDApending{2}} & \PDAcell{\star} \\
9\vphantom{\boxed{8}} & \PDAcell{\PDApending{2}} & \PDAcell{\star} & \PDAcell{\PDApending{3}} & \PDAcell{\star} & \PDAcell{\star} \\
10\vphantom{\boxed{8}} & \PDAcell{\PDAchanged{1}{\langle2\rangle_{1}}} & \PDAcell{\star} & \PDAcell{\star} & \PDAcell{\star} & \PDAcell{\PDApending{3}} \\
11\vphantom{\boxed{8}} & \PDAcell{\PDApending{2}} & \PDAcell{\star} & \PDAcell{\star} & \PDAcell{\PDApending{1}} & \PDAcell{\star} \\
12\vphantom{\boxed{8}} & \PDAcell{\PDApending{2}} & \PDAcell{\star} & \PDAcell{\PDApending{2}} & \PDAcell{\star} & \PDAcell{\star} \\
\end{array}\right]$}
\end{minipage}
\hfill
\begin{minipage}[t]{0.24\textwidth}
\centering
\parbox[t][2.6\baselineskip][t]{\linewidth}{\centering (c) $d=1$: round 2}
\resizebox{\linewidth}{!}{$
\left[\begin{array}{r|cc:ccc}
f & \PDAcell{k_1^{(1)}} & \PDAcell{k_2^{(1)}} & \PDAcell{k_1^{(2)}} &
\PDAcell{k_2^{(2)}} & \PDAcell{k_3^{(2)}}\\[2pt]
\hline
1\vphantom{\boxed{8}} & \PDAcell{\star} & \PDAcell{\PDApending{1}} & \PDAcell{\star} & \PDAcell{\star} & \PDAcell{\PDApending{1}} \\
2\vphantom{\boxed{8}} & \PDAcell{\star} & \PDAcell{\PDApending{1}} & \PDAcell{\star} & \PDAcell{\PDAchanged{2}{\langle2,2\rangle_{2}}} & \PDAcell{\star} \\
3\vphantom{\boxed{8}} & \PDAcell{\star} & \PDAcell{\PDApending{1}} & \PDAcell{\PDApending{3}} & \PDAcell{\star} & \PDAcell{\star} \\
4\vphantom{\boxed{8}} & \PDAcell{\star} & \PDAcell{\PDApending{1}} & \PDAcell{\star} & \PDAcell{\star} & \PDAcell{\PDApending{3}} \\
5\vphantom{\boxed{8}} & \PDAcell{\star} & \PDAcell{\PDApending{1}} & \PDAcell{\star} & \PDAcell{\PDApending{1}} & \PDAcell{\star} \\
6\vphantom{\boxed{8}} & \PDAcell{\star} & \PDAcell{\PDApending{1}} & \PDAcell{\PDAchanged{2}{\langle2,2\rangle_{2}}} & \PDAcell{\star} & \PDAcell{\star} \\
7\vphantom{\boxed{8}} & \PDAcell{\PDAchanged{2}{\langle2,2\rangle_{2}}} & \PDAcell{\star} & \PDAcell{\star} & \PDAcell{\star} & \PDAcell{\PDApending{1}} \\
8\vphantom{\boxed{8}} & \PDAcell{\PDApending{2}} & \PDAcell{\star} & \PDAcell{\star} & \PDAcell{\PDApending{2}} & \PDAcell{\star} \\
9\vphantom{\boxed{8}} & \PDAcell{\PDApending{2}} & \PDAcell{\star} & \PDAcell{\PDApending{3}} & \PDAcell{\star} & \PDAcell{\star} \\
10\vphantom{\boxed{8}} & \PDAcell{\langle2\rangle_{1}} & \PDAcell{\star} & \PDAcell{\star} & \PDAcell{\star} & \PDAcell{\PDApending{3}} \\
11\vphantom{\boxed{8}} & \PDAcell{\PDApending{2}} & \PDAcell{\star} & \PDAcell{\star} & \PDAcell{\PDApending{1}} & \PDAcell{\star} \\
12\vphantom{\boxed{8}} & \PDAcell{\PDApending{2}} & \PDAcell{\star} & \PDAcell{\PDApending{2}} & \PDAcell{\star} & \PDAcell{\star} \\
\end{array}\right]$}
\end{minipage}
\hfill
\begin{minipage}[t]{0.24\textwidth}
\centering
\parbox[t][2.6\baselineskip][t]{\linewidth}{\centering (d) $d=2$: rounds 1--2}
\resizebox{\linewidth}{!}{$
\left[\begin{array}{r|cc:ccc}
f & \PDAcell{k_1^{(1)}} & \PDAcell{k_2^{(1)}} & \PDAcell{k_1^{(2)}} &
\PDAcell{k_2^{(2)}} & \PDAcell{k_3^{(2)}}\\[2pt]
\hline
1\vphantom{\boxed{8}} & \PDAcell{\star} & \PDAcell{\PDApending{1}} & \PDAcell{\star} & \PDAcell{\star} & \PDAcell{\PDApending{1}} \\
2\vphantom{\boxed{8}} & \PDAcell{\star} & \PDAcell{\PDApending{1}} & \PDAcell{\star} & \PDAcell{\langle2,2\rangle_{2}} & \PDAcell{\star} \\
3\vphantom{\boxed{8}} & \PDAcell{\star} & \PDAcell{\PDApending{1}} & \PDAcell{\PDAchanged{4}{\langle2,3\rangle_{4}}} & \PDAcell{\star} & \PDAcell{\star} \\
4\vphantom{\boxed{8}} & \PDAcell{\star} & \PDAcell{\PDApending{1}} & \PDAcell{\star} & \PDAcell{\star} & \PDAcell{\PDAchanged{4}{\langle2,3\rangle_{4}}} \\
5\vphantom{\boxed{8}} & \PDAcell{\star} & \PDAcell{\PDApending{1}} & \PDAcell{\star} & \PDAcell{\PDApending{1}} & \PDAcell{\star} \\
6\vphantom{\boxed{8}} & \PDAcell{\star} & \PDAcell{\PDApending{1}} & \PDAcell{\langle2,2\rangle_{2}} & \PDAcell{\star} & \PDAcell{\star} \\
7\vphantom{\boxed{8}} & \PDAcell{\langle2,2\rangle_{2}} & \PDAcell{\star} & \PDAcell{\star} & \PDAcell{\star} & \PDAcell{\PDApending{1}} \\
8\vphantom{\boxed{8}} & \PDAcell{\PDAchanged{4}{\langle2,3\rangle_{4}}} & \PDAcell{\star} & \PDAcell{\star} & \PDAcell{\PDApending{2}} & \PDAcell{\star} \\
9\vphantom{\boxed{8}} & \PDAcell{\PDApending{2}} & \PDAcell{\star} & \PDAcell{\PDApending{3}} & \PDAcell{\star} & \PDAcell{\star} \\
10\vphantom{\boxed{8}} & \PDAcell{\langle2\rangle_{1}} & \PDAcell{\star} & \PDAcell{\star} & \PDAcell{\star} & \PDAcell{\PDApending{3}} \\
11\vphantom{\boxed{8}} & \PDAcell{\PDAchanged{3}{\langle2\rangle_{3}}} & \PDAcell{\star} & \PDAcell{\star} & \PDAcell{\PDApending{1}} & \PDAcell{\star} \\
12\vphantom{\boxed{8}} & \PDAcell{\PDApending{2}} & \PDAcell{\star} & \PDAcell{\PDApending{2}} & \PDAcell{\star} & \PDAcell{\star} \\
\end{array}\right]$}
\end{minipage}
\par\vspace{1.2ex}
\begin{minipage}[t]{0.24\textwidth}
\centering
\parbox[t][2.6\baselineskip][t]{\linewidth}{\centering (e) $d=3$: rounds 1--2}
\resizebox{\linewidth}{!}{$
\left[\begin{array}{r|cc:ccc}
f & \PDAcell{k_1^{(1)}} & \PDAcell{k_2^{(1)}} & \PDAcell{k_1^{(2)}} &
\PDAcell{k_2^{(2)}} & \PDAcell{k_3^{(2)}}\\[2pt]
\hline
1\vphantom{\boxed{8}} & \PDAcell{\star} & \PDAcell{\PDApending{1}} & \PDAcell{\star} & \PDAcell{\star} & \PDAcell{\PDApending{1}} \\
2\vphantom{\boxed{8}} & \PDAcell{\star} & \PDAcell{\PDApending{1}} & \PDAcell{\star} & \PDAcell{\langle2,2\rangle_{2}} & \PDAcell{\star} \\
3\vphantom{\boxed{8}} & \PDAcell{\star} & \PDAcell{\PDApending{1}} & \PDAcell{\langle2,3\rangle_{4}} & \PDAcell{\star} & \PDAcell{\star} \\
4\vphantom{\boxed{8}} & \PDAcell{\star} & \PDAcell{\PDApending{1}} & \PDAcell{\star} & \PDAcell{\star} & \PDAcell{\langle2,3\rangle_{4}} \\
5\vphantom{\boxed{8}} & \PDAcell{\star} & \PDAcell{\PDApending{1}} & \PDAcell{\star} & \PDAcell{\PDApending{1}} & \PDAcell{\star} \\
6\vphantom{\boxed{8}} & \PDAcell{\star} & \PDAcell{\PDApending{1}} & \PDAcell{\langle2,2\rangle_{2}} & \PDAcell{\star} & \PDAcell{\star} \\
7\vphantom{\boxed{8}} & \PDAcell{\langle2,2\rangle_{2}} & \PDAcell{\star} & \PDAcell{\star} & \PDAcell{\star} & \PDAcell{\PDApending{1}} \\
8\vphantom{\boxed{8}} & \PDAcell{\langle2,3\rangle_{4}} & \PDAcell{\star} & \PDAcell{\star} & \PDAcell{\PDApending{2}} & \PDAcell{\star} \\
9\vphantom{\boxed{8}} & \PDAcell{\PDAchanged{6}{\langle2\rangle_{6}}} & \PDAcell{\star} & \PDAcell{\PDApending{3}} & \PDAcell{\star} & \PDAcell{\star} \\
10\vphantom{\boxed{8}} & \PDAcell{\langle2\rangle_{1}} & \PDAcell{\star} & \PDAcell{\star} & \PDAcell{\star} & \PDAcell{\PDApending{3}} \\
11\vphantom{\boxed{8}} & \PDAcell{\langle2\rangle_{3}} & \PDAcell{\star} & \PDAcell{\star} & \PDAcell{\PDApending{1}} & \PDAcell{\star} \\
12\vphantom{\boxed{8}} & \PDAcell{\PDAchanged{5}{\langle2\rangle_{5}}} & \PDAcell{\star} & \PDAcell{\PDApending{2}} & \PDAcell{\star} & \PDAcell{\star} \\
\end{array}\right]$}
\end{minipage}
\hfill
\begin{minipage}[t]{0.24\textwidth}
\centering
\parbox[t][2.6\baselineskip][t]{\linewidth}{\centering (f) $d=4$: round 1}
\resizebox{\linewidth}{!}{$
\left[\begin{array}{r|cc:ccc}
f & \PDAcell{k_1^{(1)}} & \PDAcell{k_2^{(1)}} & \PDAcell{k_1^{(2)}} &
\PDAcell{k_2^{(2)}} & \PDAcell{k_3^{(2)}}\\[2pt]
\hline
1\vphantom{\boxed{8}} & \PDAcell{\star} & \PDAcell{\PDApending{1}} & \PDAcell{\star} & \PDAcell{\star} & \PDAcell{\PDApending{1}} \\
2\vphantom{\boxed{8}} & \PDAcell{\star} & \PDAcell{\PDApending{1}} & \PDAcell{\star} & \PDAcell{\langle2,2\rangle_{2}} & \PDAcell{\star} \\
3\vphantom{\boxed{8}} & \PDAcell{\star} & \PDAcell{\PDApending{1}} & \PDAcell{\langle2,3\rangle_{4}} & \PDAcell{\star} & \PDAcell{\star} \\
4\vphantom{\boxed{8}} & \PDAcell{\star} & \PDAcell{\PDApending{1}} & \PDAcell{\star} & \PDAcell{\star} & \PDAcell{\langle2,3\rangle_{4}} \\
5\vphantom{\boxed{8}} & \PDAcell{\star} & \PDAcell{\PDApending{1}} & \PDAcell{\star} & \PDAcell{\PDApending{1}} & \PDAcell{\star} \\
6\vphantom{\boxed{8}} & \PDAcell{\star} & \PDAcell{\PDApending{1}} & \PDAcell{\langle2,2\rangle_{2}} & \PDAcell{\star} & \PDAcell{\star} \\
7\vphantom{\boxed{8}} & \PDAcell{\langle2,2\rangle_{2}} & \PDAcell{\star} & \PDAcell{\star} & \PDAcell{\star} & \PDAcell{\PDApending{1}} \\
8\vphantom{\boxed{8}} & \PDAcell{\langle2,3\rangle_{4}} & \PDAcell{\star} & \PDAcell{\star} & \PDAcell{\PDAchanged{7}{\langle2\rangle_{7}}} & \PDAcell{\star} \\
9\vphantom{\boxed{8}} & \PDAcell{\langle2\rangle_{6}} & \PDAcell{\star} & \PDAcell{\PDApending{3}} & \PDAcell{\star} & \PDAcell{\star} \\
10\vphantom{\boxed{8}} & \PDAcell{\langle2\rangle_{1}} & \PDAcell{\star} & \PDAcell{\star} & \PDAcell{\star} & \PDAcell{\PDApending{3}} \\
11\vphantom{\boxed{8}} & \PDAcell{\langle2\rangle_{3}} & \PDAcell{\star} & \PDAcell{\star} & \PDAcell{\PDApending{1}} & \PDAcell{\star} \\
12\vphantom{\boxed{8}} & \PDAcell{\langle2\rangle_{5}} & \PDAcell{\star} & \PDAcell{\PDAchanged{7}{\langle2\rangle_{7}}} & \PDAcell{\star} & \PDAcell{\star} \\
\end{array}\right]$}
\end{minipage}
\hfill
\begin{minipage}[t]{0.24\textwidth}
\centering
\parbox[t][2.6\baselineskip][t]{\linewidth}{\centering (g) $d=5$: round 1}
\resizebox{\linewidth}{!}{$
\left[\begin{array}{r|cc:ccc}
f & \PDAcell{k_1^{(1)}} & \PDAcell{k_2^{(1)}} & \PDAcell{k_1^{(2)}} &
\PDAcell{k_2^{(2)}} & \PDAcell{k_3^{(2)}}\\[2pt]
\hline
1\vphantom{\boxed{8}} & \PDAcell{\star} & \PDAcell{\PDApending{1}} & \PDAcell{\star} & \PDAcell{\star} & \PDAcell{\PDApending{1}} \\
2\vphantom{\boxed{8}} & \PDAcell{\star} & \PDAcell{\PDApending{1}} & \PDAcell{\star} & \PDAcell{\langle2,2\rangle_{2}} & \PDAcell{\star} \\
3\vphantom{\boxed{8}} & \PDAcell{\star} & \PDAcell{\PDApending{1}} & \PDAcell{\langle2,3\rangle_{4}} & \PDAcell{\star} & \PDAcell{\star} \\
4\vphantom{\boxed{8}} & \PDAcell{\star} & \PDAcell{\PDApending{1}} & \PDAcell{\star} & \PDAcell{\star} & \PDAcell{\langle2,3\rangle_{4}} \\
5\vphantom{\boxed{8}} & \PDAcell{\star} & \PDAcell{\PDApending{1}} & \PDAcell{\star} & \PDAcell{\PDApending{1}} & \PDAcell{\star} \\
6\vphantom{\boxed{8}} & \PDAcell{\star} & \PDAcell{\PDApending{1}} & \PDAcell{\langle2,2\rangle_{2}} & \PDAcell{\star} & \PDAcell{\star} \\
7\vphantom{\boxed{8}} & \PDAcell{\langle2,2\rangle_{2}} & \PDAcell{\star} & \PDAcell{\star} & \PDAcell{\star} & \PDAcell{\PDApending{1}} \\
8\vphantom{\boxed{8}} & \PDAcell{\langle2,3\rangle_{4}} & \PDAcell{\star} & \PDAcell{\star} & \PDAcell{\langle2\rangle_{7}} & \PDAcell{\star} \\
9\vphantom{\boxed{8}} & \PDAcell{\langle2\rangle_{6}} & \PDAcell{\star} & \PDAcell{\PDAchanged{8}{\langle3\rangle_{8}}} & \PDAcell{\star} & \PDAcell{\star} \\
10\vphantom{\boxed{8}} & \PDAcell{\langle2\rangle_{1}} & \PDAcell{\star} & \PDAcell{\star} & \PDAcell{\star} & \PDAcell{\PDAchanged{8}{\langle3\rangle_{8}}} \\
11\vphantom{\boxed{8}} & \PDAcell{\langle2\rangle_{3}} & \PDAcell{\star} & \PDAcell{\star} & \PDAcell{\PDApending{1}} & \PDAcell{\star} \\
12\vphantom{\boxed{8}} & \PDAcell{\langle2\rangle_{5}} & \PDAcell{\star} & \PDAcell{\langle2\rangle_{7}} & \PDAcell{\star} & \PDAcell{\star} \\
\end{array}\right]$}
\end{minipage}
\hfill
\begin{minipage}[t]{0.24\textwidth}
\centering
\parbox[t][2.6\baselineskip][t]{\linewidth}{\centering (h) $d=6$: no update}
\resizebox{\linewidth}{!}{$
\left[\begin{array}{r|cc:ccc}
f & \PDAcell{k_1^{(1)}} & \PDAcell{k_2^{(1)}} & \PDAcell{k_1^{(2)}} &
\PDAcell{k_2^{(2)}} & \PDAcell{k_3^{(2)}}\\[2pt]
\hline
1\vphantom{\boxed{8}} & \PDAcell{\star} & \PDAcell{\PDApending{1}} & \PDAcell{\star} & \PDAcell{\star} & \PDAcell{\PDApending{1}} \\
2\vphantom{\boxed{8}} & \PDAcell{\star} & \PDAcell{\PDApending{1}} & \PDAcell{\star} & \PDAcell{\langle2,2\rangle_{2}} & \PDAcell{\star} \\
3\vphantom{\boxed{8}} & \PDAcell{\star} & \PDAcell{\PDApending{1}} & \PDAcell{\langle2,3\rangle_{4}} & \PDAcell{\star} & \PDAcell{\star} \\
4\vphantom{\boxed{8}} & \PDAcell{\star} & \PDAcell{\PDApending{1}} & \PDAcell{\star} & \PDAcell{\star} & \PDAcell{\langle2,3\rangle_{4}} \\
5\vphantom{\boxed{8}} & \PDAcell{\star} & \PDAcell{\PDApending{1}} & \PDAcell{\star} & \PDAcell{\PDApending{1}} & \PDAcell{\star} \\
6\vphantom{\boxed{8}} & \PDAcell{\star} & \PDAcell{\PDApending{1}} & \PDAcell{\langle2,2\rangle_{2}} & \PDAcell{\star} & \PDAcell{\star} \\
7\vphantom{\boxed{8}} & \PDAcell{\langle2,2\rangle_{2}} & \PDAcell{\star} & \PDAcell{\star} & \PDAcell{\star} & \PDAcell{\PDApending{1}} \\
8\vphantom{\boxed{8}} & \PDAcell{\langle2,3\rangle_{4}} & \PDAcell{\star} & \PDAcell{\star} & \PDAcell{\langle2\rangle_{7}} & \PDAcell{\star} \\
9\vphantom{\boxed{8}} & \PDAcell{\langle2\rangle_{6}} & \PDAcell{\star} & \PDAcell{\langle3\rangle_{8}} & \PDAcell{\star} & \PDAcell{\star} \\
10\vphantom{\boxed{8}} & \PDAcell{\langle2\rangle_{1}} & \PDAcell{\star} & \PDAcell{\star} & \PDAcell{\star} & \PDAcell{\langle3\rangle_{8}} \\
11\vphantom{\boxed{8}} & \PDAcell{\langle2\rangle_{3}} & \PDAcell{\star} & \PDAcell{\star} & \PDAcell{\PDApending{1}} & \PDAcell{\star} \\
12\vphantom{\boxed{8}} & \PDAcell{\langle2\rangle_{5}} & \PDAcell{\star} & \PDAcell{\langle2\rangle_{7}} & \PDAcell{\star} & \PDAcell{\star} \\
\end{array}\right]$}
\end{minipage}
\caption{Construction of $P_{\textnormal{com}}$ from the Cartesian product
in \eqref{eq:illustrative_cart}. Read panels (a)--(h) from left to right,
then top to bottom; $d$ indexes the all-star row tuples $(r_1,r_2)$.
Columns follow increasing cluster indices, and row $(f_1,f_2)$ has
index $f=6(f_1-1)+f_2$.
Colored boxes mark entries assigned a delivery label in the current panel;
occurrences of the same full label share a color. Unboxed labels denote
earlier replacements; gray circles mark positions not yet assigned a label.
Panel (h) retains the entries of panel (g), since $d=6$ produces no update.
Removing the row labels from panel (h) gives \eqref{P_com}.}
\label{fig:pda_steps}
\end{figure*}
\endgroup

Block $1$ occupies columns $1$--$2$, and block $2$ columns $3$--$5$.
Columns $1$ and $3$ have Reduce functions; columns $4$ and $5$ are
connected without demand, and column $2$ is disconnected.
In a displayed position $(f,k)$, $f$ is the row number and $k$ the column
number. Cartesian product coordinates $(f_1,f_2)$ correspond to row
$f=6(f_1-1)+f_2$; for example, coordinates $(2,4)$ give row $10$.

The tuples $(r_1,r_2)$ are processed as
$(1,1)$, $(1,2)$, $(1,3)$, $(2,1)$, $(2,2)$, $(2,3)$, indexed by $d=1,\ldots,6$.
Rounds restart at each tuple; label subscripts continue across tuples.

For $d=1$, $\boldsymbol r=(1,1)$, symbol $2$ has row $1$ as an all-star
row in both basic PDAs. Its row sets are
\begin{IEEEeqnarray*}{rCl}
\operatorname{Row}_1(2)\times\mathcal Z_2&=&\{2\}\times\{1,4\},\\
\mathcal Z_1\times\operatorname{Row}_2(2)&=&\{1\}\times\{2,6\}.
\end{IEEEeqnarray*}
Block $2$ gives one copy at $(2,4)$ and $(6,3)$.
In block $1$, fixing $f_2=4$ or $1$ gives copies at $(10,1)$ or $(7,1)$,
respectively, taken in that order. The list lengths in blocks $1$ and $2$
are $2$ and $1$, respectively.

Round $1$ selects block $1$, the longer list, and labels $(10,1)$ with
$\langle2\rangle_1$ (panel~(b)). With one copy now remaining in each
block, round $2$ selects both and labels $(7,1)$, $(2,4)$, and $(6,3)$
with $\langle2,2\rangle_2$ (panel~(c)). Entry $(2,4)$ stays in the
selected copy alongside $(6,3)$, although column $4$ has no demand.

For $d=3$, $\boldsymbol r=(1,3)$, the eligible block-$2$ symbol $1$
occurs only in columns $4$ and $5$, without demand. Thus
$\mathcal H_{2,(1,3)}=\emptyset$. The block-$1$ copies receive
$\langle2\rangle_5$ at $(12,1)$ and $\langle2\rangle_6$ at $(9,1)$
in separate rounds (panel~(e)).

Tuple $d=2$ follows the same two-round pattern (panel~(d)). For $d=4,5$,
the block-$1$ copies lie in disconnected column $2$, leaving only block
$2$ (panels~(f)--(g)). Neither block contributes a copy for $d=6$
(panel~(h)). The round counts are $2,2,2,1,1,0$, producing eight labels
and the final $P_{\textnormal{com}}$ in \eqref{P_com}.

\subsection{Delivery-Aware Reduce Function Reassignment}
\label{sec:reassignment}
When an initial worker disconnects, its Reduce functions must be reassigned.
Process these $|\mathcal D_{\rm ini}|$ functions in increasing function index.
For each candidate cluster $i\in[C]$ with a connected worker, tentatively
assign the next function to a connected worker with the fewest assigned
Reduce functions. If several workers have this minimum count, choose the
one with the smallest global index.
This gives the candidate Reduce function assignment vector
$\boldsymbol\gamma^{(i)}$. Workers with this minimum number of assigned Reduce functions give the same
simulated communication load by Proposition~\ref{prop:balanced}.
For $\boldsymbol\gamma^{(i)}$, the active-worker set in episode
$u\in[U^{(i)}]$, where
$U^{(i)}=\max_{k\in[K]}\gamma_k^{(i)}$, is
\begin{IEEEeqnarray}{rCl}
\mathcal A_u^{(i)}=\{k\in\bar{\mathcal D}:\gamma_k^{(i)}\ge u\}.
\end{IEEEeqnarray}
By \eqref{number_elem}, the number of adjustment rounds for index tuple
$\boldsymbol r\in\mathcal R$ can be calculated directly as
\begin{IEEEeqnarray}{rCl}
O_{\max}^{(u,\boldsymbol r,i)}
&=&\max_{c\in[C]}V_c\,\mathbf1\{\mathcal A_u^{(i)}\cap \mathcal K_c^\star(r_c)\ne\emptyset\}.
\label{eq:candidate_rounds}
\end{IEEEeqnarray}
Use the total number of rounds as the simulated communication load for
candidate comparison:
\begin{IEEEeqnarray}{rCl}
 L_{\rm sim}^{(i)}
 =\sum_{u=1}^{U^{(i)}}\sum_{\boldsymbol r\in\mathcal R}
 O_{\max}^{(u,\boldsymbol r,i)}.
\end{IEEEeqnarray}
We assign the Reduce function to the worker in the cluster with the least simulated communication load
\begin{IEEEeqnarray}{rCl}
 i^*\in\operatorname{arg\,min}_{\substack{i\in[C]:\\
 \mathcal K_i\cap\bar{\mathcal D}\ne\emptyset}}L_{\rm sim}^{(i)},
\end{IEEEeqnarray}

\begin{proposition}
\label{prop:balanced}
Fix a candidate cluster $i\in[C]$ with a connected worker.
Among connected workers in this cluster, assigning an abandoned Reduce
function to a worker with the fewest assigned Reduce functions does not
increase $L_{\rm sim}^{(i)}$ compared with assigning it to a worker with
more assigned functions. Workers with this minimum count give the same
$L_{\rm sim}^{(i)}$.
\end{proposition}
\begin{IEEEproof}
The proof is given in Appendix~\ref{app:proof-balanced}.
\end{IEEEproof}

The rule evaluates at most $C$ candidates per abandoned function; it does
not guarantee a globally optimal assignment of all abandoned functions.
\subsection{Delivery}
\label{sec:delivery}
For each episode $u\in[U]$, delete the columns of disconnected workers
from $P_{\textnormal{com}}^u$. Keep the same array notation and index the
remaining columns by the workers' global indices.
Process each distinct non-star element $s$, ignoring $\circ$.
Suppose $s$ occurs $n_s$ times at positions $(b_i,k_i)$, $i\in[n_s]$.
Its active receivers are indexed by
\begin{IEEEeqnarray}{rCl}
\mathcal J_s^u&=&\{i\in[n_s]:\gamma_{k_i}^u=1\}.
\label{eq:ns_definition}
\end{IEEEeqnarray}
Only symbols with $\mathcal J_s^u\ne\emptyset$ require delivery.
By the cross-star property in Proposition~\ref{prop:communication_pda},
the corresponding subarray is
\begin{IEEEeqnarray*}{rCl}
\left[\begin{array}{c|cccc}
 & k_1 & k_2 & \cdots & k_{n_s}\\ \hline
 b_1 & s & \star & \cdots & \star\\
 b_2 & \star & s & \cdots & \star\\
 \vdots & \vdots & \vdots & \ddots & \vdots\\
 b_{n_s} & \star & \star & \cdots & s
\end{array}\right].
\end{IEEEeqnarray*}
For $i\in\mathcal J_s^u$, let $x_i$ concatenate the IVs for function
$\phi_{q_{k_i,u}}$ from the $N/\FCart$ files in batch $b_i$, forming one
packet of $NT/\FCart$ bits. Worker $k_j$ caches every $x_i$ with
$i\in\mathcal J_s^u\setminus\{j\}$ by the cross-star property.
We deliver these packets using the following three cases.

Case 1: A connected worker $k^\star\notin\{k_1,\ldots,k_{n_s}\}$
caches all occurrence batches, that is,
$P_{\textnormal{com}}^u[b_i,k^\star]=\star$ for every $i\in[n_s]$.
The corresponding subarray is
\begin{IEEEeqnarray*}{rCl}
\left[\begin{array}{c|ccccc}
 & k_1 & k_2 & \cdots & k_{n_s} & k^\star\\ \hline
 b_1 & s & \star & \cdots & \star & \star\\
 b_2 & \star & s & \cdots & \star & \star\\
 \vdots & \vdots & \vdots & \ddots & \vdots & \vdots\\
 b_{n_s} & \star & \star & \cdots & s & \star
\end{array}\right].
\end{IEEEeqnarray*}
Worker $k^\star$ multicasts
$\mathcal X=\bigoplus_{i\in\mathcal J_s^u}x_i$.
Each active receiver $k_j$, $j\in\mathcal J_s^u$, cancels the cached
summands to recover $x_j$. The cost is one packet, or $NT/\FCart$ bits.

Case 2: Case 1 does not apply and $\mathcal J_s^u\subsetneq[n_s]$.
Choose a worker $k_j$ with $j\in[n_s]\setminus\mathcal J_s^u$ as the
sender. It has no demand in this episode and caches every requested packet,
so it multicasts
\begin{IEEEeqnarray}{rCl}
\mathcal X_j&=&\bigoplus_{i\in\mathcal J_s^u}x_i.
\end{IEEEeqnarray}
Each active receiver $k_i$, $i\in\mathcal J_s^u$, recovers its packet as
\begin{IEEEeqnarray}{rCl}
x_i&=&\mathcal X_j\bigoplus
\left(\bigoplus_{\ell\in\mathcal J_s^u\setminus\{i\}}x_\ell\right),
\end{IEEEeqnarray}
using its cached packets. The cost is one packet, or $NT/\FCart$ bits.

Case 3: Case 1 does not apply and $\mathcal J_s^u=[n_s]$.
Assume the file-availability condition of Proposition~\ref{prop:feasibility}.
If $n_s=1$, a connected worker caching batch $b_1$ would provide a
Case-1 sender. Hence $n_s\geq2$ in this case.
For each $j\in[n_s]$, split $x_j$ into $n_s-1$ equal-length subpackets
$\{x_j^{(i)}:i\in[n_s]\setminus\{j\}\}$. Worker $k_i$, $i\in[n_s]$,
multicasts
\begin{IEEEeqnarray}{rCl}
\mathcal X_i=\bigoplus_{j\in[n_s],j\ne i}x_j^{(i)}.
\end{IEEEeqnarray}
Receiver $k_j$ recovers $x_j^{(i)}$ from each $\mathcal X_i$, $i\ne j$,
by canceling the other cached summands, and obtains all subpackets of $x_j$.
Symbol extension ensures an integral split when needed.
Each of the $n_s$ messages has $1/(n_s-1)$ of a packet's length,
so the cost is $n_s/(n_s-1)$ packets, or $n_sNT/((n_s-1)\FCart)$ bits.

For each episode $u\in[U]$ and $j\in[3]$, define
\begin{IEEEeqnarray}{rCl}
\mathcal S_j^u&\triangleq&\{s:\text{$s$ is sent by Case $j$ in episode $u$}\}.
\label{eq:case_symbol_sets}
\end{IEEEeqnarray}
These sets contain only non-star symbols with $\mathcal J_s^u\ne\emptyset$;
$\star$ and $\circ$ are excluded. Thus $|\mathcal S_j^u|$ is the number of
delivery symbols handled by Case $j$.
For $s\in\mathcal S_3^u$, $n_s$ counts its occurrences after deleting the
columns of disconnected workers and equals its number of active receivers.
Each symbol in $\mathcal S_1^u\cup\mathcal S_2^u$ costs one packet, while
$s\in\mathcal S_3^u$ costs $n_s/(n_s-1)$ packets.
Since each packet contains $NT/\FCart$ bits, summing over all episodes and
normalizing by $NQT$, with $Q=K_0$, gives the system communication load
\begin{IEEEeqnarray}{rCl}
L_{\rm dyn}&=&\sum_{u=1}^{U}
\frac{|\mathcal S_1^u|+|\mathcal S_2^u|
+\sum_{s\in\mathcal S_3^u}\frac{n_s}{n_s-1}}
{K_0\FCart}.
\label{eq:scheme_communication_load}
\end{IEEEeqnarray}

\textit{Completeness of Delivery:}
By Proposition~\ref{prop:communication_pda}, every missing batch requested by an
active worker is represented by exactly one delivery symbol in its episode.
Cases 1--3 exhaust the possible transmitter and demand configurations, and their
decoding rules in Section~\ref{sec:delivery} recover every requested packet.
Processing all symbols and episodes therefore supplies all missing IVs;
together with its local IVs, each worker can compute every assigned output.

\subsection{Computational Complexity}
\label{sec:complexity}
We analyze the computational overhead of the proposed dynamic CDC scheme,
excluding file transfer and storage I/O. Its construction processes
$|\mathcal R|=\prod_{c\in[C]}\binom{K_c}{t_c}$ index tuples per episode.
Placement-array construction and occurrence indexing take
$O(C\FCart K)$ operations; recording file assignments takes
$O(N\tau)$ operations.
Sharing the array and indices across episodes and recording only changes
requires $O(\FCart K)$ auxiliary storage, excluding cached files and IVs.

Reassignment examines at most
$|\mathcal D_{\rm ini}|CU_{\max}|\mathcal R|$ episode-tuple combinations,
with at most $\tau$ activity checks each; $U_{\max}\leq Q$ is the largest
candidate episode count. With
$V_c=(K_c-t_c)\prod_{j\in[C]\setminus\{c\}}t_j$, constructing the
Communication Arrays requires at most $U|\mathcal R|\max_cV_c$ rounds, each costing
$O(K\tau)$ operations including transmitter selection. XOR encoding and
decoding require $O(NQT\tau)$ bit operations. The total computational overhead is
\begin{IEEEeqnarray}{l}
O\!\Bigl(C\FCart K+NQT\tau\nonumber\\
\quad{}+|\mathcal R|\tau\bigl[|\mathcal D_{\rm ini}|CU_{\max}
+UK\max_cV_c\bigr]\Bigr).
\end{IEEEeqnarray}

\section{Conclusion}
\label{sec:conclusion}
We proposed a dynamic CDC scheme for worker arrivals and disconnections that
preserves surviving caches and Reduce function assignments. Cartesian product placement
incorporates arriving storage, while delivery-aware reassignment coordinates
the resulting IV demands with coded Shuffle delivery. We characterized the
exact communication load for arbitrary feasible disconnection sets and
established a tight worst-case disconnection tolerance in terms of file
replication. In the absence of disconnections, we established a 
converse bound for one-shot communication. Relative to the
minimum load over all feasible placements in this class, the scheme
achieves approximation factors of $2$ with arrivals and $4$ without.
Extending these guarantees to arbitrary feasible disconnection patterns and
general CDC Shuffle schemes remains open.

\appendices

\section{Proof of Corollary~\ref{pro2}}
\label{app:proof-no-disconnection}
There is one episode, with $\mathcal C_{\rm act}^1=\mathcal C_{\rm ini}$.
For each tuple $\boldsymbol r\in\mathcal R$, active cluster $c$ contributes
$V_c=\rho_c\prod_{j\in[C]}t_j$ elements, giving
$\rho_{\max}\prod_{j\in[C]}t_j$ rounds per tuple and
$\rho_{\max}\FCart$ rounds overall. Each packet has $NT/\FCart$ bits.

By \eqref{eq:compatible_group_coordinates}, occurrence rows contributed by
cluster $i$ are stars in $\mathcal K_j^\star(r_j)$ for every $j\ne i$.
A selected cluster-$j$ component has exactly these occurrence columns,
as shown in Section~\ref{sec:comPDA}. Two facts follow.

\emph{(i) An omitted cluster yields a Case-1 round.}
If a round omits cluster $j$, every column in $\mathcal K_j^\star(r_j)$
is a star in all occurrence rows and lies outside the symbol's occurrence
columns. Since all workers are connected and a sender needs no Reduce
function, any such column provides a Case-1 transmitter at a cost of one packet.

\emph{(ii) The number of balanced rounds is at most
$\min_{c\in\mathcal C_{\rm ini}}V_c$.}
The merge rule selects all active clusters only when their positive remaining
counters are equal; call these rounds balanced. Every other round omits an
active cluster and satisfies (i). Each active cluster is selected in exactly $V_c$
rounds, so the number of balanced rounds cannot exceed the smallest $V_c$.

With arrivals, every new cluster is omitted in every round, so (i) makes
all rounds cost one packet. Normalizing the $\rho_{\max}\FCart$ packets by
$K_0\FCart$ gives \eqref{load_new}. If $\tau=1$, every symbol has one
occurrence and is sent uncoded at the same packet cost, so \eqref{load_new}
also follows.

Without arrivals and with $\tau>1$, every worker has one function, so Case 2
cannot occur. Case-3 rounds select all clusters, have $n_s=\tau$, and are
balanced. By (ii), their total number is at most $\rho_{\min}\FCart$;
each costs $\tau/(\tau-1)$ packets instead of one. Hence
\begin{IEEEeqnarray}{rCl}
L_{\rm dyn}&=&\frac{\rho_{\max}}{K_0}
+\frac{|\mathcal S_3^1|}{K_0\FCart(\tau-1)},
\label{load_no_new_exact}
\end{IEEEeqnarray}
and $0\leq|\mathcal S_3^1|\leq\rho_{\min}\FCart$ proves
\eqref{load_no_new}.

When $C=C_{\rm ini}=1$, each basic-PDA symbol comes from an MN symbol
indexed by a $(t_1+1)$-subset. Its $t_1$ occurrence columns omit one
worker of that subset, which caches every occurrence row. Thus every round
is Case 1, giving $\rho_{\max}/K_0$ and meeting Lemma~\ref{lem:lower}.

When $C=C_{\rm ini}\geq2$ and $\rho_{\max}=\rho_{\min}$, every round
merges all clusters. Any all-star column of the merged symbol in cluster $j$ must belong to
$\mathcal K_j^\star(r_j)$, since another cluster contributes a copy of row
$r_j$. But these are precisely the occurrence columns of cluster $j$'s
component, each with a non-star entry in an occurrence row. Thus no
all-star column of the merged symbol exists, every round is Case 3, and the upper endpoint
of \eqref{load_no_new} is attained.

\section{Proof of Theorem~\ref{thm:multicast}}
\label{app:proof-multicast}
Split transmissions into atomic XOR packets of equal-length IV pieces without
changing the total message length. Fix a packet $\mathcal X$ of $\ell$ bits
and a cluster $c$. Let $\mathcal J_c(\mathcal X)\subseteq[K_c]$ index the connected
workers decoding a missing piece from $\mathcal X$, and let
$h=|\mathcal J_c(\mathcal X)|$. For each $b\in\mathcal J_c(\mathcal X)$,
let $n_b\notin\mathcal M_b^{(c)}$ be the file associated with its decoded piece.

For any distinct $a,b\in\mathcal J_c(\mathcal X)$, worker $k_a^{(c)}$ cancels
the term recovered by worker $k_b^{(c)}$ using its cache alone. Hence
$n_b\in\mathcal M_a^{(c)}$. It follows that at least the other $h-1$
receivers of cluster $c$ cache file $n_b$, and
therefore
\begin{IEEEeqnarray*}{rCl}
r_{c,n_b}&\geq&h-1.
\end{IEEEeqnarray*}
Assign every requested bit associated with file $n$ in cluster $c$ the weight
$1/(r_{c,n}+1)$. When $h\geq1$, the total weighted number of requested bits
delivered to that cluster by $\mathcal X$ is at most
\begin{IEEEeqnarray}{rCl}
\ell\sum_{b\in\mathcal J_c(\mathcal X)}\frac{1}{r_{c,n_b}+1}
&\leq&\ell\frac{h}{h}=\ell.
\end{IEEEeqnarray}
The bound also holds for $h=0$. Summing over packets, with $d_{c,n}$ missing
$T$-bit IVs associated with file $n$ in cluster $c$, gives
\begin{IEEEeqnarray}{rCl}
T\sum_{n=1}^{N}\frac{d_{c,n}}{r_{c,n}+1}
&\leq& LNK_0T.
\end{IEEEeqnarray}
This holds for every cluster $c$; dividing by $NK_0T$ and taking the maximum
over $c$ proves \eqref{eq:placement_lb}.

Without disconnections, each initial cluster satisfies
$d_{c,n}=K_c-r_{c,n}$ and $N^{-1}\sum_n r_{c,n}=t_c$.
The function
\begin{IEEEeqnarray}{rCl}
f_c(x)&=&\frac{K_c-x}{x+1}=\frac{K_c+1}{x+1}-1
\end{IEEEeqnarray}
is convex on $[0,K_c]$. Jensen's inequality therefore gives
\begin{IEEEeqnarray}{rCl}
\frac{1}{N}\sum_{n=1}^{N}\frac{K_c-r_{c,n}}{r_{c,n}+1}
&\geq&\frac{K_c-t_c}{t_c+1}.
\end{IEEEeqnarray}
Substitution into \eqref{eq:placement_lb} proves
\eqref{eq:placement_free_lb}, which remains valid after minimizing over
$\mathcal P$ because its right-hand side is independent of $\mathcal M$.

\section{Proof of Theorem~\ref{thm:gap}}
\label{app:proof-gap}
Choose an initial cluster attaining $\rho_{\max}$ with replication
$t_\star$. Theorem~\ref{thm:multicast} gives, for every feasible $\mathcal M$,
\begin{IEEEeqnarray}{rCl}
L_{\textnormal{one-shot}}^*(\mathcal M)&\geq&
\frac{\rho_{\max}t_\star}{K_0(t_\star+1)}.
\label{eq:placement_opt_lb}
\end{IEEEeqnarray}
\looseness=-1
With arrivals, dividing $L_{\rm dyn}=\rho_{\max}/K_0$ by this bound
gives $1+1/t_\star\leq2$. Without arrivals, $\tau=t_0$ and
\eqref{load_no_new} yields the second branch of
\eqref{eq:placement_gap}. Since $t_\star\geq1$,
$\rho_{\min}\leq\rho_{\max}$, and $t_0\geq2$, its two factors
are each at most two. The converse is uniform over $\mathcal M$, so the same
constants apply after minimizing over $\mathcal P$.

\looseness=-1
For the remaining case $\tau=1$, there is one initial cluster with
$t_1=1$ and no arrivals. Every feasible placement then stores
each file at exactly one worker, since its total cache capacity is $N$
files. Lemma~\ref{lem:lower} gives the load lower bound
$(K_0-1)/K_0$, which is attained by Corollary~\ref{pro2}. Thus the scheme
is optimal in this case, and the constants in
\eqref{eq:unconditional_gap} also hold when $\tau=1$.

\looseness=-1
For $(K_1,t_1)=(2,1)$ and $(K_2,t_2)=(4,3)$, pairing the initial workers'
missing IVs achieves load $1/4$, optimal by \eqref{eq:placement_free_lb}.
Corollary~\ref{pro2} gives $L_{\rm dyn}=1/2$, exactly twice the optimum.
Thus, the factor $2$ with arrivals is tight.

\section{Proof of Proposition~\ref{prop:feasibility}}
\label{app:proof-feasibility}
In cluster $c$, the workers caching a batch form a $t_c$-subset of
$[K_c]$. Among these $t_c$ workers, at most $\delta_c$ disconnect, so
every batch has at least $\max\{t_c-\delta_c,0\}$ surviving copies in
that cluster. This minimum is attained by choosing a $t_c$-subset
containing all disconnected workers if $\delta_c<t_c$, or contained in
the disconnected workers if $\delta_c\geq t_c$.
The MN star pattern includes every $t_c$-subset, and the Cartesian
product includes every combination of these subsets across clusters.
Thus one batch attains all cluster-wise minima simultaneously. Every
batch contains $N/\FCart\geq1$ files, which proves
\eqref{eq:min_surviving_copies}.

All files have a surviving copy precisely when the right-hand side of
\eqref{eq:min_surviving_copies} is positive, equivalently when
$\delta_c<t_c$ for some $c$. Under this condition there is a connected
worker to receive each abandoned function, and the exhaustive delivery
cases in Section~\ref{sec:delivery} recover every missing IV.
Conversely, a file with no surviving copy has no available IVs and cannot
be recovered by the Shuffle transmissions.
If $|\mathcal D|<\tau$, the inequalities $\delta_c\geq t_c$ cannot hold
for all clusters, proving \eqref{eq:disconnection_guarantee}. Finally,
disconnecting exactly the $\tau$ workers caching any fixed batch removes
all copies of its files, proving tightness.

\section{Proof of Proposition~\ref{prop:communication_pda}}
\label{app:proof-communication-pda}
Among the occurrences of one selected copy of a basic symbol, the cross-star
condition follows from the corresponding basic PDA.
For occurrences contributed by copies from distinct blocks
$i$ and $j$, the coordinate rule
\eqref{eq:compatible_group_coordinates} puts each occurrence in an
all-star row of the other block's symbol. Both cross entries are
therefore stars. Their rows are distinct because the block-$i$
occurrence lies in an occurrence row of $s_i$, whereas the block-$j$
occurrence has a copy of an all-star row of $s_i$ as its $i$-th
coordinate. Their columns are distinct because they belong to
different blocks. Each round produces a distinct merged element, so the
cross-star condition holds throughout the array.

Every non-star position in block $i$ requested by an active worker determines
a unique tuple $\boldsymbol r$ through its symbol's designated all-star row and
its other row coordinates.
Its worker lies in $\mathcal K_i^\star(r_i)$, so its symbol copy remains
in $\mathcal H_{i,\boldsymbol r}$. By \eqref{eq:compatible_group_coordinates},
the position belongs to exactly one such copy. Exhausting the lists thus labels
every requested position exactly once; copies omitted from the construction
have no active receiver.

\section{Proof of Proposition~\ref{prop:balanced}}
\label{app:proof-balanced}
For an active-worker set $\mathcal A\subseteq\bar{\mathcal D}$, define the one-episode simulated communication load
\begin{IEEEeqnarray}{rCl}
G(\mathcal A)=\sum_{\boldsymbol r\in\mathcal R}
\max_{c\in[C]}V_c\mathbf 1\{\mathcal A\cap \mathcal K_c^\star(r_c)\neq\emptyset\},
\end{IEEEeqnarray}
By \eqref{number_elem}, the simulated communication load is
$\sum_{u=1}^{U}G(\mathcal K_{\rm act}^u)$.
Here the active sets refer to the current assignment before adding the new function,
with $\mathcal K_{\rm act}^u=\emptyset$ for $u>\max_k\gamma_k$.

For a fixed $\boldsymbol r\in\mathcal R$, adding a worker of cluster $c\in[C]$
incurs zero marginal cost when cluster $c$ is already represented or when that
worker lies outside $\mathcal K_c^\star(r_c)$; otherwise, its marginal cost equals the positive
part of $V_c$ minus the largest weight already represented. Since this marginal
cost cannot increase as the active set grows, $G$ is monotone and submodular.

Now fix distinct workers $p,q\in\mathcal K_i\cap\bar{\mathcal D}$ with
current Reduce function counts $a=\gamma_p$ and $b=\gamma_q$, where $0\le a<b$.
Exchanging $p$ and $q$ permutes the cluster-$i$ subset indices and leaves $G$
unchanged. Assigning the new function to $p$ or $q$ affects only episode
$a+1$ or $b+1$, respectively. The two marginal costs are thus
$\Delta_p=G(\mathcal K_{\rm act}^{a+1}\cup\{p\})-G(\mathcal K_{\rm act}^{a+1})$ and
$\Delta_q=G(\mathcal K_{\rm act}^{b+1}\cup\{q\})-G(\mathcal K_{\rm act}^{b+1})$.
Since $\mathcal K_{\rm act}^{b+1}\subseteq \mathcal K_{\rm act}^{a+1}$, and since $q\in \mathcal K_{\rm act}^{a+1}$ while neither $p$ nor
$q$ belongs to $\mathcal K_{\rm act}^{b+1}$, exchanging their labels in the first
marginal-cost expression changes its base set to
$(\mathcal K_{\rm act}^{a+1}\setminus\{q\})\cup\{p\}$. This set contains
$\mathcal K_{\rm act}^{b+1}$ and does not contain $q$. Therefore, invariance under
the label exchange and submodularity give $\Delta_p\leq\Delta_q$.
This proves the comparison for $L_{\rm sim}^{(i)}$. If $a=b$, then neither $p$ nor $q$ belongs to
$\mathcal K_{\rm act}^{a+1}$; exchanging their labels therefore shows that the
two marginal costs coincide.

\bibliographystyle{IEEEtran}
\bibliography{ref}

@ARTICLE{10086693,
  author={Yang, Ting and Wan, Kai and Cheng, Minquan and Qiu, Robert Caiming and Caire, Giuseppe},
  journal={IEEE Transactions on Information Theory}, 
  title={Multiple-antenna placement delivery array for cache-aided {MISO} systems}, 
  year={2023},
  volume={69},
  number={8},
  pages={4855-4868},
  doi={10.1109/TIT.2023.3262935}}

@article{wu2024coded,
  title={Coded computing for multi-cluster distributed computations},
  author={Wu, Youlong and Li, Chenglin and Hu, Haoyang and Song, Xiyu and Ma, Shuai and Shi, Yuanming},
  journal={IEEE Transactions on Communications},
  volume={73},
  number={2},
  pages={1114--1127},
  year={2025},
  doi={10.1109/TCOMM.2024.3446641},
  publisher={IEEE}
}

@INPROCEEDINGS{11195520,
  author={Zhong, Xi and Lu, Samuel and Kliewer, Jörg and Ji, Mingyue},
  booktitle={2025 IEEE International Symposium on Information Theory (ISIT)}, 
  title={Dual-{Lagrange} encoding for storage and download in elastic computing for resilience}, 
  year={2025},
  volume={},
  number={},
  pages={1-6},
  doi={10.1109/ISIT63088.2025.11195520}}

@ARTICLE{7935426,
  author={Li, Songze and Yu, Qian and Maddah-Ali, Mohammad Ali and Avestimehr, A. Salman},
  journal={IEEE/ACM Transactions on Networking}, 
  title={A Scalable Framework for Wireless Distributed Computing}, 
  year={2017},
  volume={25},
  number={5},
  pages={2643-2654},
  doi={10.1109/TNET.2017.2702605}}

@ARTICLE{6763007,
  author={Maddah-Ali, Mohammad Ali and Niesen, Urs},
  journal={IEEE Transactions on Information Theory}, 
  title={Fundamental Limits of Caching}, 
  year={2014},
  volume={60},
  number={5},
  pages={2856-2867},
  doi={10.1109/TIT.2014.2306938}}

@INPROCEEDINGS{8613522,
  author={Yan, Qifa and Tang, Xiaohu and Chen, Qingchun},
  booktitle={2018 IEEE Information Theory Workshop (ITW)}, 
  title={Placement Delivery Array and Its Applications}, 
  year={2018},
  volume={},
  number={},
  pages={1-5},
  doi={10.1109/ITW.2018.8613522}}

@INPROCEEDINGS{8437603,
  author={Yan, Qifa and Wigger, Michèle and Yang, Sheng},
  booktitle={2018 IEEE International Symposium on Information Theory (ISIT)}, 
  title={Placement Delivery Array Design for Combination Networks with Edge Caching}, 
  year={2018},
  volume={},
  number={},
  pages={1555-1559},
  doi={10.1109/ISIT.2018.8437603}}

@INPROCEEDINGS{9815628,
  author={Salehi, MohammadJavad and Parrinello, Emanuele and Mahmoodi, Hamidreza Bakhshzad and T\"{o}lli, Antti},
  booktitle={2022 Joint European Conference on Networks and Communications \& 6G Summit (EuCNC/6G Summit)}, 
  title={Low-Subpacketization Multi-Antenna Coded Caching for Dynamic Networks}, 
  year={2022},
  volume={},
  number={},
  pages={112-117},
  doi={10.1109/EuCNC/6GSummit54941.2022.9815628}
}

@ARTICLE{6807823,
  author={Maddah-Ali, Mohammad Ali and Niesen, Urs},
  journal={IEEE/ACM Transactions on Networking}, 
  title={Decentralized Coded Caching Attains Order-Optimal Memory-Rate Tradeoff}, 
  year={2015},
  volume={23},
  number={4},
  pages={1029-1040},
  doi={10.1109/TNET.2014.2317316}}

@ARTICLE{7973185,
  author={Yan, Qifa and Cheng, Minquan and Tang, Xiaohu and Chen, Qingchun},
  journal={IEEE Transactions on Information Theory}, 
  title={On the Placement Delivery Array Design for Centralized Coded Caching Scheme}, 
  year={2017},
  volume={63},
  number={9},
  pages={5821-5833},
  doi={10.1109/TIT.2017.2725272}}

@article{wang2023placement,
  title={Placement delivery array construction via {Cartesian} product for coded caching},
  author={Wang, Jinyu and Cheng, Minquan and Wan, Kai and Caire, Giuseppe},
  journal={IEEE Transactions on Information Theory},
  volume={69},
  number={12},
  pages={7602--7626},
  year={2023},
  publisher={IEEE},
  doi={10.1109/TIT.2023.3320252}
}

@article{wang2022coded,
  title={Coded distributed computing with pre-set data placement and output functions assignment},
  author={Wang, Yuhan and Wu, Youlong},
  journal={IEEE Transactions on Information Theory},
  volume={71},
  number={3},
  pages={2195--2217},
  year={2025},
  doi={10.1109/TIT.2025.3528083}
}

@article{dean2008mapreduce,
  title={{MapReduce}: Simplified data processing on large clusters},
  author={Dean, Jeffrey and Ghemawat, Sanjay},
  journal={Communications of the ACM},
  volume={51},
  number={1},
  pages={107--113},
  year={2008},
  publisher={ACM New York, NY, USA}
}

@inproceedings{zaharia2010spark,
  title={{Spark}: Cluster computing with working sets},
  author={Zaharia, Matei and Chowdhury, Mosharaf and Franklin, Michael J and Shenker, Scott and Stoica, Ion},
  booktitle={2nd USENIX workshop on hot topics in cloud computing (HotCloud 10)},
  year={2010}
}

@article{li2017fundamental,
  title={A fundamental tradeoff between computation and communication in distributed computing},
  author={Li, Songze and Maddah-Ali, Mohammad Ali and Yu, Qian and Avestimehr, A Salman},
  journal={IEEE Transactions on Information Theory},
  volume={64},
  number={1},
  pages={109--128},
  year={2018},
  publisher={IEEE},
  doi={10.1109/TIT.2017.2756959}
}

@ARTICLE{9448271,
  author={Woolsey, Nicholas and Chen, Rong-Rong and Ji, Mingyue},
  journal={IEEE Transactions on Communications}, 
  title={A New Combinatorial Coded Design for Heterogeneous Distributed Computing}, 
  year={2021},
  volume={69},
  number={9},
  pages={5672-5685},
  doi={10.1109/TCOMM.2021.3087628}}

@article{reisizadeh2019coded,
  title={Coded computation over heterogeneous clusters},
  author={Reisizadeh, Amirhossein and Prakash, Saurav and Pedarsani, Ramtin and Avestimehr, Amir Salman},
  journal={IEEE Transactions on Information Theory},
  volume={65},
  number={7},
  pages={4227--4242},
  year={2019},
  publisher={IEEE}
}

@article{xu2021new,
  title={New Results on the Computation-Communication Tradeoff for Heterogeneous Coded Distributed Computing},
  author={Xu, Fan and Shao, Shuo and Tao, Meixia},
  journal={IEEE Transactions on Communications},
  volume={69},
  number={4},
  pages={2254--2270},
  year={2021},
  publisher={IEEE}
}

@article{chen2023optimality,
  title={On the Optimality of Data Exchange for Master-Aided Edge Computing Systems},
  author={Chen, Haoning and Long, Junfeng and Ma, Shuai and Tang, Mingjian and Wu, Youlong},
  journal={IEEE Transactions on Communications},
  volume={71},
  number={3},
  pages={1364--1376},
  year={2023},
  publisher={IEEE}
}

@inproceedings{chen2021coded,
  title={Coded computing for master-aided distributed computing systems},
  author={Chen, Haoning and Wu, Youlong},
  booktitle={2020 IEEE Information Theory Workshop (ITW)},
  pages={1--5},
  year={2021},
  organization={IEEE}
}

@INPROCEEDINGS{CEC, 
  author={Yang, Yaoqing and Interlandi, Matteo and Grover, Pulkit and Kar, Soummya and Amizadeh, Saeed and Weimer, Markus},
  booktitle={2019 IEEE Int. Symp. Inf. Theory (ISIT)}, 
  title={Coded elastic computing}, 
  year={2019},
  volume={},
  number={},
  pages={2654-2658},
  publisher={IEEE},
  doi={10.1109/ISIT.2019.8849212}
}

\end{document}